\documentclass[10pt,a4paper]{article}

\usepackage{amsmath,amsgen,latexsym}
\usepackage{amstext,amssymb,amsfonts,latexsym}
\usepackage{theorem}
\usepackage{pifont}
\usepackage{graphicx}

 \newcommand{\bs}{\bigskip}
 \newcommand{\ms}{\medskip}
 \newcommand{\n}{\noindent}
 \newcommand{\s}{\smallskip}
 \newcommand{\hs}[1]{\hspace*{ #1 mm}}
 \newcommand{\vs}[1]{\vspace*{ #1 mm}}

 \newcommand{\setempty}{\varnothing}
 
 \newcommand{\nat}{\mathbb{N}}
 
 \newcommand{\integer}{\mathbb{Z}}

 \newcommand{\co}{\mathrm{co}\mbox{-}}

 \newcommand{\AAA}{{\cal A}}

 \newcommand{\CC}{{\cal C}}
 \newcommand{\FF}{{\cal F}}
 \newcommand{\DD}{{\cal D}}

 \newcommand{\KK}{{\cal K}}
 \newcommand{\LL}{{\cal L}}

 \newcommand{\MM}{{\cal M}}

 \newcommand{\PP}{{\cal P}}

 \newcommand{\dl}{\mathrm{L}}
 \newcommand{\nl}{\mathrm{NL}}
 
 \newcommand{\np}{\mathrm{NP}}

 \newcommand{\poly}{\mathrm{poly}}

 \newcommand{\cfl}{\mathrm{CFL}}

\theoremstyle{plain}
\theoremheaderfont{\bfseries}
 \newtheorem{theorem}{Theorem}[section]
 \newtheorem{lemma}[theorem]{Lemma}
 \newtheorem{proposition}[theorem]{{\bf Proposition}}
 \newtheorem{corollary}[theorem]{Corollary}

 {\theorembodyfont{\rmfamily}  \newtheorem{definition}[theorem]{Definition}}
 {\theorembodyfont{\rmfamily} }
 {\theorembodyfont{\rmfamily} }

 \newenvironment{proofof}[1]{\vspace*{5mm} \par \noindent
         {\bf Proof of #1.\hs{2}}}{\hfill$\Box$ \vspace*{3mm}}

 \newtheorem{yclaim}[theorem]{Claim}

 \newenvironment{proof}{\par \noindent
            {\bf Proof. \hs{2}}}{\hfill$\Box$ \vspace*{3mm}}

 \newcommand{\floors}[1]{\lfloor #1 \rfloor}
 \newcommand{\pair}[1]{\langle #1 \rangle}

\newcommand{\ignore}[1]{}

 \newcommand{\unary}{\mathrm{unary}}

 \newcommand{\oned}{1\mathrm{D}}

 \newcommand{\twod}{2\mathrm{D}}
 
 \newcommand{\twon}{2\mathrm{N}}

 \newcommand{\para}{\mathrm{para}\mbox{-}}

 \newcommand{\phsp}{\mathrm{PHSP}}

 \newcommand{\twou}{2\mathrm{U}}

\newcommand{\logdcfl}{\mathrm{LOGDCFL}}
\newcommand{\logcfl}{\mathrm{LOGCFL}}

\newcommand{\ptime}{\mathrm{ptime}\mbox{-}}

\newcommand{\up}{\mathrm{UP}}
\newcommand{\ul}{\mathrm{UL}}

\newcommand{\twodpd}{2\mathrm{DPD}}
\newcommand{\twonpd}{2\mathrm{NPD}}
\newcommand{\twonct}{2\mathrm{NCT}}
\newcommand{\twoupd}{2\mathrm{UPD}}
\newcommand{\twouct}{2\mathrm{UCT}}
\newcommand{\twodct}{2\mathrm{DCT}}

\newcommand{\twonpdct}{2\mathrm{NPDCT}}
\newcommand{\twoupdct}{2\mathrm{UPDCT}}
\newcommand{\twodpdct}{2\mathrm{DPDCT}}

\newcommand{\PHSP}{\mathrm{PHSP}}
\newcommand{\logucfl}{\mathrm{LOGUCFL}}

\newcommand{\LGOOD}{\mathrm{LGOOD}}

\begin{document}

\pagestyle{plain}
\pagenumbering{arabic}
\setcounter{page}{1}
\setcounter{footnote}{0}

\begin{center}
{\Large Unambiguous and Co-Nondeterministic Computations of Finite Automata and Pushdown Automata Families and the Effects of Multiple Counters}\footnote{A preliminary report \cite{Yam24} appeared in the Proceedings of
the 18th Annual Conference on Theory and Applications of Models of Computation (TAMC 2024), Hong Kong, China, May 13-15, 2024, Lecture Notes in Computer Science, vol. 14637, pp. 14--25, Springer, 2024.}
\bs\ms\\

{\sc Tomoyuki Yamakami}\footnote{Present Affiliation: Faculty of Engineering, University of Fukui, 3-9-1 Bunkyo, Fukui 910-8507, Japan} \bs\\
\end{center}
\ms

\begin{abstract}
Nonuniform families of polynomial-size finite automata and pushdown automata respectively have strong connections to nonuniform-NL and nonuniform-LOGCFL.
We examine the behaviors of unambiguous and co-nondeterministic  computations produced by such families of automata operating multiple counters, where a counter is a stack using
only a single non-bottom symbol.
As immediate consequences, we obtain various collapses of the complexity classes of
families of promise problems solvable by finite and pushdown automata families when all valid instances are limited to either polynomially long strings or unary strings.
A key technical ingredient of our proofs is an inductive counting of reachable vertices of each computation graph of finite and pushdown automata that operate multiple counters simultaneously.
\end{abstract}

\sloppy
\section{Background and Challenging Questions}\label{sec:introduction}

This section will provide background knowledge on the topics of this work, raise important open questions, and overview their partial solutions.

\subsection{Two Important Open Questions on Nonuniform (Stack-)State Complexity Classes}\label{sec:state-complexity}

\emph{Nondeterministic computation} has played an important role in computational complexity theory as well as automata theory.
Associated with such computation, there are two central and crucial questions to resolve. (i) Can any co-nondeterministic computation be simulated on an  appropriate nondeterministic machine?
(ii) Can any nondeterministic machine be made unambiguous?
In the polynomial-time setting, these questions correspond to the  famous $\np=?\co\np$ and $\np=?\up$ questions.
In this work, we attempt to resolve these questions in the setting of nonuniform polynomial state complexity classes.

We quickly review the origin and the latest progress of the study of nonuniform state complexity classes.
Apart from a standard uniform model of finite automata, Berman and Lingas \cite{BL77} and Sakoda and Sipser \cite{SS78} considered,  as a ``collective'' model of computations, using nonuniform families of finite automata indexed by natural numbers and, in particular, they studied the computational power of these families of finite automata having \emph{polynomial size} (i.e., having polynomially many inner states).
Unlike Boolean circuit families, these automata are allowed to take ``arbitrarily'' long inputs to read and process.
As underlying models of polynomial-size machines, Sakoda and Sipser focused on \emph{two-way deterministic finite automata} (or 2dfa's, for short) as well as \emph{two-way nondeterministic finite automata} (or 2nfa's).
They introduced the complexity classes, dubbed as $\twod$ and $\twon$, which consist of all families of ``promise'' problems\footnote{A \emph{promise (decision) problem} over alphabet $\Sigma$ is a pair $(A,R)$ satisfying that  $A,R\subseteq \Sigma^*$ and $A\cap R=\setempty$. A \emph{language} $L$ over $\Sigma$ can be identified with a unique promise problem having the form $(L,\Sigma^*-L)$.} solvable respectively by nonuniform families of polynomial-size 2dfa's and 2nfa's running in polynomial time.
As their natural extensions, nonuniform families of  \emph{two-way deterministic pushdown automata} (or 2dpda's) and {two-way nondeterministic pushdown automata} (or 2npda's) whose runtimes are bounded by suitable polynomials have also been studied lately in \cite{Yam21,Yam23b}.
Similarly to $\twod$ and $\twon$, these pushdown automata models induce two corresponding complexity classes of promise problem families, denoted respectively by $\twodpd$ and $\twonpd$ \cite{Yam21,Yam23b}.
Since an introduction of nonuniform polynomial-size finite automata families,
various machine types (such as deterministic, nondeterministic, alternating,  probabilistic, and quantum) have been studied in depth   \cite{Gef12,Kap09,Kap12,Kap14,KP15,Yam19a,Yam21,Yam22a,Yam22b,Yam23b}.
When underlying finite and pushdown automata for $\twon$ and $\twonpd$ are unambiguous (at least on all valid instances), we use the notations $\twou$ \cite{Kap09,Kap12} and $\twoupd$ \cite{Yam22b}, respectively.

For nonuniform families of finite and pushdown automata, nevertheless, the aforementioned two central open questions (i)--(ii) correspond to the $\twon=?\twou$,  $\twon=?\co\twon$, $\twonpd =? \twoupd$, and $\twonpd =? \co\twonpd$ questions.
Unfortunately, these four equalities are not known to hold at this moment.
It is therefore desirable to continue the study on the behaviors of finite and pushdown automata families in order to deepen our understandings of these machine families and to eventually resolve those four central questions.

An importance of polynomial-size finite automata families comes from their close connection to decision problems in the nonuniform variants of the log-space complexity classes $\dl$ and $\nl$, when all promised (or valid) instances given to underlying finite automata are limited to polynomially long strings (where this condition is referred to as  a  \emph{polynomial ceiling}) \cite{Kap14}, or when all valid instances are limited to unary strings  \cite{KP15}.
In a similar fashion, when instances are restricted to polynomial ceilings, $\twodpd$ and $\twonpd$ are closely related to
the nonuniform versions of $\logdcfl$ and $\logcfl$ \cite{Yam21},
where $\logcfl$ (resp., $\logdcfl$) is the collection of all languages log-space many-one reducible to  context-free (resp., deterministic context-free) languages \cite{Sud78}.
These strong correspondences to standard complexity classes provide one of the good reasons to investigate the fundamental properties of various types of nonuniform finite and pushdown automata families in hopes of achieving a better understanding of parallel complexity classes,
such as $\dl$, $\nl$, $\logdcfl$, and $\logcfl$, in the nonuniform setting.

\subsection{New Challenges and Main Contributions}

The primary purpose of this work is to present ``partial'' solutions to the two important questions (i)--(ii) raised in Section \ref{sec:state-complexity} associated particularly with $\twon$ and $\twonpd$ by studying
the computational power of nonuniform families of polynomial-size finite and pushdown automata in depth.

In the course of our study, we further look into the key role of ``counters'', each of
which is essentially a stack manipulating only a single symbol, say, ``$1$'' except for the bottom marker $\bot$. Since the total number of $1$s in a counter can be viewed as a natural number, the counter is able to ``count'', as its name suggests.
We supplement multiple counters to the existing models of nonuniform finite and pushdown automata families. This is the first time that counters are formally treated in the context of nonuniform {(stack-)state} complexity.
We remark that, for short runtime computation, counters are significantly weaker\footnote{With exponential overhead, 2-counter automata can simulate a Turing machine  \cite{Min67}} in functionality than full-scale stacks.
Even though, the proper use of counters can help us not only trace the tape head location but also count the number of steps.

By appending multiple counters to finite and pushdown automata, we obtain the  machine models of \emph{counter automata} and \emph{counter pushdown automata}.
Two additional complexity classes,
$\twonct_k$ and
$\twonpdct_k$ are naturally obtained by taking families of polynomial-size nondeterministic counter automata and counter pushdown automata operating $k$ counters. It is possible to reduce the number of counters ($\twonct_k=\twonct_4$). An appropriate use of a stack further reduces the number of counters to $3$ ($\twonpdct_k=\twonpdct_3$).

The use of multiple counters makes it possible to show in Section \ref{sec:reachability} that $\twonct_4$ and $\co\twonct_4$ coincide. By another direct simulation, we will show that $\co\twonpd$ and $\twonpd$ also coincide. These results further lead to the equivalences between $\co\twon$ and $\twon$ and between $\co\twonpd$ and $\twonpd$ in Section \ref{sec:complementary} when all promise problem families are restricted to having polynomial ceilings.
Our results are briefly summarized in Fig.~\ref{fig:hierarchy}, in which the suffix ``$/\poly$''
refers to the polynomial ceiling restriction
and the suffix ``$/\mathrm{unary}$'' refers to the restriction to unary input strings.
To obtain some of the equalities in the figure, we will exploit a close connection between \emph{parameterized decision problems} and families of promise problems, which was first observed in \cite{Yam22a} and then fully developed in \cite{Yam21,Yam22b,Yam23b}.


\begin{figure}[t]
\centering
\includegraphics*[height=5.0cm]{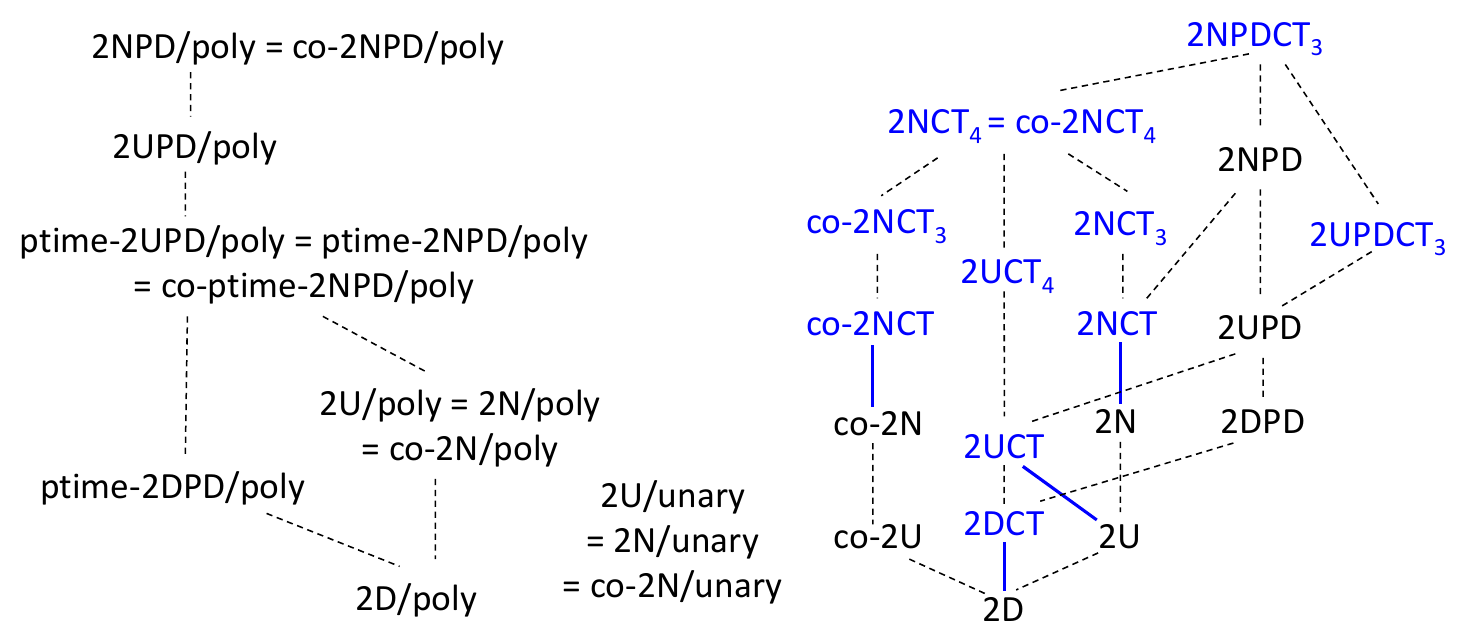}
\caption{Containments among nonuniform polynomial {(stack-)state} complexity classes shown in this work. A blue solid line indicates proper inclusion whereas a black dotted line does simple inclusion (not known to be proper).
Remark that the collapse $\twou/\poly=\twon/\poly$
comes from \cite{Yam22b} and $\twou/\unary = \twon/\unary = \co\twon/\unary$ are drawn from \cite{GMP07,GP11} in Section \ref{sec:unary-inputs}.
}\label{fig:hierarchy}
\end{figure}


\section{Preliminaries: Notions and Notation}\label{sec:notion-notation}

We briefly explain fundamental notions and notation used in the rest of this work.

\subsection{Numbers, Languages, and Pushdown Automata}\label{sec:numbers}

The set of all \emph{natural numbers} (including $0$) is denoted $\nat$ and the positive-integer set $\nat-\{0\}$ is expressed as $\nat^{+}$. Given two integers $m$ and $n$ with $m\leq n$, the notation $[m,n]_{\integer}$ denotes the \emph{integer set} $\{m,m+1,m+2,\ldots,n\}$. As a special case, we write $[n]$ for $[1,n]_{\integer}$ when $n\in\nat^{+}$. In this work, all  \emph{polynomials} must have nonnegative coefficients and all \emph{logarithms} are taken to the base $2$ with the notation $\log{0}$ being treated as $0$.
The \emph{power set} of a set $S$ is denoted $\PP(S)$.

An \emph{alphabet} is a finite nonempty set of ``symbols'' or ``letters''. Given an alphabet $\Sigma$ and a number $n\in\nat$, the notation $\Sigma^n$ (resp., $\Sigma^{\leq n}$) denotes the set of all strings over $\Sigma$ of length exactly $n$ (resp., at most $n$). The \emph{empty string} is always denoted $\varepsilon$.

As a basic machine model, we use \emph{two-way nondeterministic finite automata} (or 2nfa's, for short) that make neither  $\varepsilon$-moves\footnote{An \emph{$\varepsilon$-move} (or an $\varepsilon$-transition) of an automaton refers to an action of the automaton, which makes a transition without reading any input symbol.} nor stationary moves.
This means that input tape heads of 2nfa's always move to adjacent tape cells without stopping at any tape cells.
Formally, a 2nfa $M$ is of the form $(Q,\Sigma,\{\rhd,\lhd\}, \delta,q_0, Q_{acc}, Q_{rej})$ with a finite set $Q$ of inner states, an (input) alphabet $\Sigma$, the initial (inner) state $q_0$, and a set $Q_{acc}$ (resp., $Q_{rej}$) of accepting (resp., rejecting) states satisfying $Q_{acc},Q_{rej}\subseteq Q$ and $Q_{acc}\cap Q_{reje}=\setempty$.
A \emph{halting (inner) state} is either an accepting state or a rejecting state.
We write $Q_{halt}$ for $Q_{acc}\cup Q_{rej}$.
Given a finite automaton $M$, the \emph{state complexity} $sc(M)$ of $M$ is the total number of inner states used for $M$; namely, $sc(M)=|Q|$.
A 2nfa starts with an input $x$, surrounded by two endmarkers $\rhd,\lhd$ written on an input tape, makes a series of transitions by applying $\delta$, and terminates with halting states. Remember that $M$ is allowed to make \emph{$\varepsilon$-moves} (or $\varepsilon$-transitions).
The machine is said to \emph{accept} $x$ if there exists an accepting computation path and \emph{reject} if all computation paths are rejecting.

Another important machine model is \emph{two-way nondeterministic pushdown automata} (or 2npda's) $N$ over alphabet $\Sigma$ with
a stack alphabet $\Gamma$ including the stack's bottom marker $\bot$. Notice that the bottom marker is neither popped nor
pushed into any non-bottom cell.
The \emph{stack-state complexity} $ssc(N)$ of $N$ denotes the product $|Q|\cdot |\Gamma^{\leq e}|$, which turns out to be a useful complexity measure \cite{Yam21,Yam23b}, where $Q$ is a set of inner states and $e$ is the \emph{push size} (i.e., the maximum length of pushed strings into the stack at any single push operation).
We remark that the one-way restriction of 2npda's defines \emph{context-free languages}. The complexity class $\cfl$ is composed of all such context-free languages.

We also use \emph{two-way nondeterministic auxiliary pushdown automata} (or aux-2npda's), each of which is equipped with a read-only input tape, a stack, and a (two-way) auxiliary (work) tape \cite{Coo71}. We introduce the notation $\mathrm{NAuxPDA,\!\!TISP}(t(n),s(n))$ for the collection of all languages recognized by aux-2npda's running in $O(t(n))$ time using $O(s(n))$ work space. Sudborough \cite{Sud78} demonstrated that $\mathrm{NAuxPDA,\!\!TISP}(n^{O(1)},\log{n})$ coincides with $\logcfl$ (i.e., the closure of $\cfl$ under log-space many-one reductions).


\begin{figure}[t]
\centering
\includegraphics*[height=5.0cm]{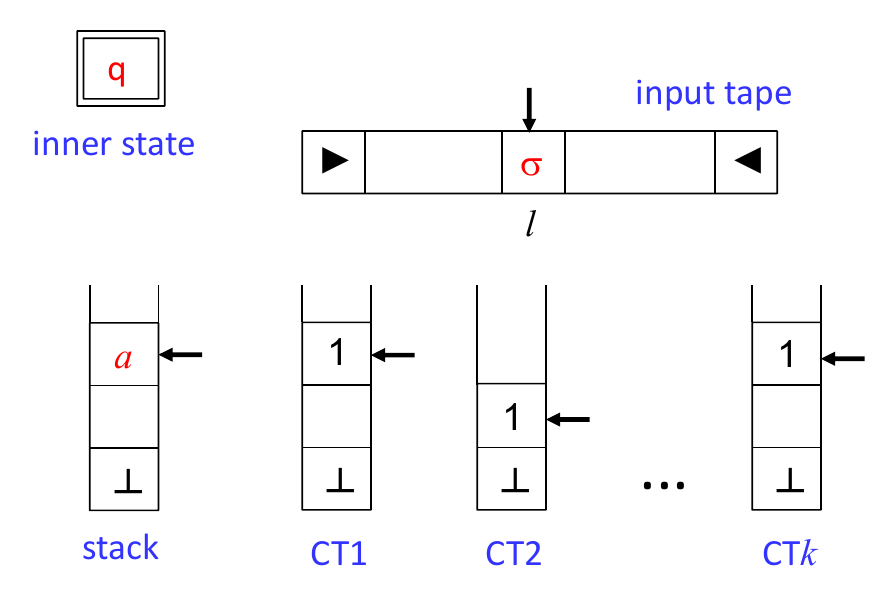}
\caption{A two-way $k$-counter pushdown automaton. Counters are indexed from CT1 to CT$k$. An input-tape head is scanning the $l$th tape cell, which holds an input symbol $\sigma$. The topmost stack cell holds a stack symbol $a$.}\label{fig:2npda}
\end{figure}


\subsection{Advice and Advised Machines}\label{sec:advised-machine}

A key notion of this work is advice and advised machines.
A piece of \emph{advice} is given to an underlying machine in the form of \emph{advice string}, which is initially written on a (two-way) read-only advice tape.

Those machines equipped with advice are briefly called \emph{advised aux-2npda's}. Such an advised aux-2npda $M$ with an advice function $h$ is said to \emph{recognize} a language $L$ over alphabet $\Sigma$ if, for any actual input $x\in\Sigma^*$, $M$ reads $(x,h(|x|))$ written on an input tape and an advice tape and eventually halts by accepting (resp., rejecting) the input if $x\in L$ (resp., $x\in\overline{L}$), where $\overline{L}=\Sigma^*-L$.
By supplementing advice tapes to underlying aux-2npda's, we can define $\mathrm{NAuxPDA,\!\!TISP}(t(n),s(n))/\poly$ from $\mathrm{NAuxPDA,\!\!TISP}(t(n),s(n))$. Similarly, we obtain $\mathrm{UAuxPDA,\!\!TISP}(t(n),s(n))/\poly$ using advised aux-2upda's.

To improve the readability, in the rest of this work, we succinctly denote $\mathrm{NAuxPDA,\!\!TISP}(n^{O(1)},\log{n})/\poly$ by $\logcfl/\poly$ and  $\mathrm{UAuxPDA,\!\!TISP}(n^{O(1)},\log{n})/\poly$ by $\logucfl/\poly$.

\subsection{Promise Problems and Nonuniform Families}\label{sec:promise-problem}

Unlike the notion of languages, \emph{promise problems} over alphabet $\Sigma$ are formally of the form $(A,R)$ satisfying that $A,R\subseteq \Sigma^*$ and $A\cap R=\setempty$.
We say that a 1nfa (1ncta, 1npda, or 1npdcta) $M$ \emph{solves} $(A,R)$ if (i) for any $x\in A$, $M$ accepts $x$ (i.e., there exists an accepting computation path of $M$ on $x$) and (ii) for any $x\in R$, $M$ rejects $x$ (i.e., all computation paths of $M$ on $x$ are rejecting). Any string in $A\cup R$ is said to be \emph{promised} or \emph{valid}. Since we do not impose any further condition on all strings outside of $A\cup R$, it suffices to focus only on the promised strings in our later discussion.

For any nondeterministic machine models discussed in Section \ref{sec:numbers}, a machine is said to be  \emph{unambiguous} if it has at most one accepting computation path on each promised instance. For other instances, there is no restriction on the number of accepting/rejecting computation paths.

Throughout this work, we consider a ``family'' $\LL$ of promise problems $(L_n^{(+)},L_n^{(-)})$ over a common fixed alphabet $\Sigma$ indexed by natural numbers $n\in\nat$. Such a family $\LL$ is said to have a \emph{polynomial ceiling} if there exists a polynomial $p$ such that $L_n^{(+)}\cup L_n^{(-)}\subseteq \Sigma^{\leq p(n)}$ holds for all indices  $n\in\nat$.
Given a complexity class $\CC$ of families of promise problems, if we restrict our attention to only promise problem families in $\CC$ having a polynomial ceiling, then we obtain the subclass of $\CC$, expressed as $\CC/\poly$.
Moreover, when all promise problems are restricted to the ones over unary alphabets, we obtain the subclass $\CC/\mathrm{unary}$. Those exotic notations come from \cite{Kap09,Kap12} and are adopted in \cite{KP15,Yam19a,Yam21,Yam22a,Yam22b,Yam23b}.

A family $\MM= \{M_n\}_{n\in\nat}$ of 2nfa's (resp., 2npda's) over alphabet $\Sigma$ is of \emph{polynomial size} if there exists a polynomial $p$ satisfying $sc(M_n)\leq p(n)$ (resp., $ssc(M_n)\leq p(n)$) for all $n\in\nat$.
Similar notions are definable for other machine models, such as 2dfa's, 2ncta's, 2npda's, 2npdcta's, and 2dpdcta's.

Given a family $\LL=\{(L_n^{(+)},L_n^{(-)})\}_{n\in\nat}$ of promise problems, a family $\MM=\{M_n\}_{n\in\nat}$ of nondeterministic machines, and a polynomial $p$, we say that $M_n$
\emph{solves $(L_n^{(+)},L_n^{(-)})$ within time $p(n,|x|)$} if (1) for any $x\in L_n^{(+)}$, there exists an accepting computation path of $M_n$ on $x$  and (2) for any $x\in L_n^{(-)}$, all halting computation paths of $M_n$ on $x$ are rejecting.
Moreover, $\MM$ is said to
\emph{solve $\LL$ in polynomial time} if there is a polynomial $p$ such that, for all indices $n\in\nat$, $M_n$ solves $(L_n^{(+)},L_n^{(-)})$ within time $p(n,|x|)$.

We define $\twon$ as the collection of all families
of promise problems solvable by
nonuniform families of polynomial-size 2nfa's.
It is important to note that, as shown by Geffert et al. \cite{GMP07}, in the case of 2dfa's and 2nfa's, placing the ``polynomial time'' requirement does not change the above definitions of $\twod$ and $\twon$.
By further replacing 2nfa's with 2dfa's, 2dpda's, and 2npda's, we respectively obtain $\twod$, $\twodpd$, and $\twonpd$. The notation $\co\LL$ denotes the family $\{(L_n^{(-)},L_n^{(+)})\}_{n\in\nat}$ obtained from $\LL$.
Given a complexity class $\CC$ of promise problem families, such as $\twon$ and $\twonpd$, $\co\CC$ expresses the class $\{\co\LL\mid \LL\in \CC\}$. It follows that $\twod=\co\twod$ and $\twodpd=\co\twodpd$ by swapping between accepting and rejecting states of underlying machines.
In addition, the use of unambiguous 2nfa's and unambiguous 2npda's introduces the complexity classes $\twou$ and $\twoupd$, respectively.

\section{Power and Limitation of Multiple Counters}\label{sec:multiple-counters}

In this work, we intend to further equip finite and pushdown automata with multiple ``counters'' in order to enhance their computational power.
A \emph{counter} is a special kind of (pushdown) stack whose  alphabet consists only of a single symbol, say, ``$1$''  except for $\bot$; namely, $\Gamma=\{1,\bot\}$. The use of such counters were found in the literature (e.g., \cite{Min67}).

\subsection{Machines Equipped with Multiple Counters}

In this work, we freely equip multiple counters to finite automata and pushdown automata. For clarity, these machines are respectively called \emph{counter automata} and \emph{counter pushdown automata} in this work.
We conveniently abbreviate a \emph{two-way nondeterministic counter automaton} as a 2ncta and a \emph{two-way nondeterministic counter pushdown automaton} as a 2npdcta.
Formally, $k$-counter 2ncta and 2npdcta are respectively of the form $(Q,\Sigma,\{\rhd,\lhd\}, k, \{1,\bot\}, \delta,q_0,Q_{acc},Q_{rej})$ and  $(Q,\Sigma,\{\rhd,\lhd\}, \Gamma, k, \{1,\bot\}, \delta,q_0, \bot, Q_{acc},Q_{rej})$, where $\{1,\bot\}$ is the counter alphabet and $k$ indicates the number of counters in use.

We further expand the complexity classes $\twon$ and $\twonpd$ by modifying their underlying 2nfa's and 2npda's.
For each number $k\in\nat^{+}$, we define $\twonct_k$ and $\twonpdct_k$ using $k$-counter 2ncta's and $k$-counter 2npdcta's,\footnote{We remark that, with the use of multiple counters, it is possible to force 2ncta's and 2npdcta's to halt within polynomial time on \emph{all computation paths}.} respectively.
When $k=1$, we tend to drop the subscript ``$k$'' and write $\twonct$ and $\twonpd$ instead of $\twonct_1$ and $\twonpd_1$.
We also consider the multi-counter variants of $\twou$ and $\twoupd$, denoted by $\twouct_k$ and $\twoupdct_k$, respectively.
It then follows that
$\twod\subseteq \twodpd \subseteq \twoupd \subseteq \twoupdct$ and $\twou\subseteq \twon\subseteq \twonct \subseteq \twonpd \subseteq \twonpdct$.
Refer to Fig.~\ref{fig:hierarchy} for various inclusion relations among those complexity classes.

Here is one simple example of exhibiting the power of counters.

\begin{lemma}
$\twod\neq \twodct$, $\twou\neq \twouct$, and $\twon\neq \twonct$.
\end{lemma}

\begin{proof}
Our goal is to prove that $\twodct\nsubseteq \twon$ because we immediately obtain a contradiction from any of the following equalities: $\twod=\twodct$, $\twou=\twouct$, and $\twon=\twonct$.
Let us consider $\LL_{EQ}=\{(L_n^{(+)},L_n^{(-)})\}_{n\in\nat}$ of promise problems defined by $L_n^{(+)} = \{w\# w\mid w\in\{0,1\}^{m(n)}\}$, where $m(n)=2^{n}$, and $L_n^{(-)}=\{0,1\}^{2m(n)+1} - L_n^{(+)}$ for each number $n\in\nat$. We then claim that $\LL_{EQ}$ is in $\twodct$ by storing the information on the tape head location using a counter. More precisely, the following procedure recognizes the promise problem $(L_n^{+)},L_n^{(-)})$.
\vs{-1}
\begin{quote}
Let an input $x$ have the form $u\# v$. Now, assume that a tape head is at cell $0$ and that CT1 contains $1^m$ for $m\geq0$. (1) Push ``1'' into CT1. (2) Move a tape head to cell $m+1$ by decrementing CT1 one by one.
(3) Read a tape symbol, say, $a$ at cell $m+1$, remember it in inner states, and then increment CT1 by moving the tape head back to cell $0$.
(4) Move the tape head back to $\#$.
(a) If $a=\#$, check if $|v|=m+1$ by decrementing CT1. If so, accept and halt. Otherwise, reject. (b) If $a\neq \#$, then
move the tape head further by decrementing CT1. During this process, if the tape head reaches $\lhd$, then reject. Assume otherwise.
(5) At the time when CT1 is empty, read a tape symbol, say, $b$. If $a\neq b$, then reject and halt. Otherwise, increment CT1 by moving the tape head back to $\#$. Move it further to cell $0$ and then go to (1).
\end{quote}
\vs{-1}

On the contrary, we intend to prove that $\LL_{EQ}$ is not in $\twon$.
Since $\twon\subseteq 2^{\oned}$ \cite{Kap09,Kap12}, it actually suffices to show that $\LL_{EQ}\notin 2^{\oned}$. We intend to employ a simple counting argument.
Assume that there exists a family $\{M_n\}_{n\in\nat}$ of polynomial-size 2dfa's $M_n$ solving $\LL$ with $Q_n$, where $|Q_n|$ is upper-bounded by $2^{p(n)}$ for a suitable polynomial $p$.
Choose a sufficiently large number $n$ in $\nat$ satisfying that $p(n)< m(n)$.
Let $w\#w$ denote an arbitrary input with $|w|=m(n)$. We call by a \emph{block} each series of tape cells containing $w$ in $w\#w$. Whenever the tape head crosses over the middle separator $\#$, $M_n$'s tape head carries only one inner state from the left block to the right one.
The total number of distinct strings $w$ of length $m(n)$ is exactly $2^{m(n)}$ whereas the total number of inner states carried over $\#$ between the two blocks is at most $|Q_n|$.
Since $|Q_n|<2^{m(n)}$, there exist two distinct strings $w_1$ and $w_2$ of length $m(n)$ for which $M_n$ enters the same inner state, say, $q$ after reading $w_1$ and $w_2$ in the first blocks of $w_1\#w_1$ and $w_2\#w_2$.   Since $w_1\# w_1$ and $w_2\# w_2\in L_n^{+}$, the string $w_1\# w_2$ is also accepted by $M_n$. This is a contradiction.
\end{proof}

\subsection{Reducing the Number of Counters in Use}\label{sec:reducing-counter}

It is possible to reduce the number of counters in use on multi-counter automata and multi-counter pushdown automata.
Minsky \cite{Min67} earlier demonstrated how to simulate a Turing machine on a 2-counter automaton with exponential overhead. Since we cannot use the same simulation technique due to its large overhead, we need to take another, more direct approach toward $\twonct_k$ and $\twonpdct_k$.
Even without the requirement of polynomial ceiling, it is possible in general to reduce the number of counters in use down to ``4'' for 2ncta's and ``3'' for 2npdcta's as shown below.

\begin{proposition}\label{reducing-counter}
For any constants $k,k'\in\nat$ with $k\geq4$ and $k'\geq3$,  $\twonct_k=\twonct_{4}$ and $\twonpdct_{k'}=\twonpdct_{3}$.
The same holds for the deterministic case.
\end{proposition}

A core of the proof of Proposition \ref{reducing-counter} is the following lemma on the simulation of every pair of counters by a single counter with the heavy use of an appropriately defined ``pairing''  function. Let $\MM=\{M_n\}_{n\in\nat}$ denote any polynomial-size family of 2ncta's or of 2npdcta's running in time, in particular, $(n|x|)^t$ for a fixed constant $t\in\nat^{+}$. This $\MM$ satisfies the following lemma. See also \cite{Yam23c} for the reduction of the number of counters in the uniform setting.

\begin{lemma}\label{claim-reduction}
There exists a fixed deterministic procedure by which any single move of push/pop operations of two counters of $M_n$ can be simulated by a series of operations with one counter with the help of 3 extra counters. These extra 3 counters are emptied after each simulation and thus they are reusable for any other purposes. If we freely use a stack during this simulation procedure, then we need only two extra counters instead of three. The state complexity of the procedure is $n^{O(1)}$.
\end{lemma}

For readability, we postpone the proof of Lemma \ref{claim-reduction} until Appendix and continue the proof of Proposition \ref{reducing-counter}.


\vs{-2}
\begin{proofof}{Proposition \ref{reducing-counter}}
We first look into the case of $\twonct_k$ for every index $k\geq4$ and wish to prove that $\twonct_{k} \subseteq  \twonct_{4}$.

Let $\MM=\{M_n\}_{n\in\nat}$ denote any nonuniform family of polynomial-size $k$-counter 2ncta's running in polynomial time, where $M_n$ has the form $(Q_n,\Sigma,k,\{1,\bot\}, \delta_n,q_{0,n}, Q_{acc,n},Q_{rej,n})$. Take two polynomials $p_1$ and $p_2$ such that $|Q_n|\leq p_1(n)$ and $M_n$ halts within $p_2(n,|x|)$ steps for any number $n\in\nat$ and any string $x\in\Sigma^*$.

In what follows, we fix $n$ and $x$ arbitrarily. We abbreviate $p_2(n,p_1(n))+1$ as $p$ and consider the pairing function $\pair{i_1,i_2}_{p}$ defined as $\pair{i_1,i_2}_p = i_1\cdot p + i_2$ for any $i_1,i_2\in[0,p-1]_{\integer}$.

We first group together $2\cdot\floors{k/2}$ counters into pairs and apply Lemma \ref{claim-reduction} to simulate each pair of counters by a single counter with the help of three extra reusable counters, say, CT1--CT3. This process successfully eliminates $\floors{k/2}$
counters except for CT1--CT3.
We repeat this process until there remains one counter other than CT1--CT3. Since there are the total of four counters left unremoved, this shows that $\twonct_k\subseteq \twonct_4$.

Next, we intend to show that $\twonpdct_{k'}\subseteq \twonpdct_3$ for any $k'\geq3$. In this case, we follow the same argument as described above
by removing all counters except for one counter and CT1--CT3. Finally, four counters are left unremoved. As shown in Lemma \ref{claim-reduction}, we can further reduce the number of counters to three. This is because we can utilize a stack, which is originally provided to an underlying pushdown automaton.
\end{proofof}


It is not clear that ``4''  and ``3''  are the smallest numbers supporting Proposition  \ref{reducing-counter} for 2ncta's and 2npdcta's, respectively. We may conjecture that $\twonct_i \neq  \twonct_{i+1}$ and $\twonpdct_j \neq \twonpdct_{j+1}$ for all $i\in[3]$ and $j\in[2]$.

\subsection{Reachability by Multi-Counter Automata Families}\label{sec:reachability}

Let us consider the question raised in Section \ref{sec:state-complexity} on the closure property under complementation,
namely, the $\twon=?\co\twon$ question. Unfortunately, we do not know its answer at this moment.
With the presence of ``counters'', however, it is possible to
provide a complete solution to this question.

\begin{theorem}\label{counter-four}
For any constant $k\geq4$, $\co\twonct_k \subseteq \twonct_4$. Thus, $\twonct_4 = \co\twonct_4$ follows.
\end{theorem}

A key to the proof of Theorem \ref{counter-four} is the following lemma.

\begin{lemma}\label{complement-upper-bound}
For any constant $k\in\nat^{+}$, $\co\twonct_k\subseteq \twonct_{5k+13}$.
\end{lemma}

Theorem \ref{counter-four} follows, as shown below, from this lemma with the help of Proposition \ref{reducing-counter}.

\vs{-2}
\begin{proofof}{Theorem \ref{counter-four}}
The proof of Theorem \ref{counter-four} is described as follows. By Proposition \ref{reducing-counter}, it suffices to consider the case of $k=4$.
We then obtain $\co\twonct_4\subseteq \twonct_{33}$ by Lemma  \ref{complement-upper-bound}. Proposition \ref{reducing-counter} again leads to  $\twonct_{33}\subseteq \twonct_{4}$. Therefore, $\co\twonct_4\subseteq \twonct_4$ follows. By taking the ``complementation'' of the both sides of this inclusion, we also obtain $\twonct_4\subseteq \co\twonct_4$. The second part of the theorem is thus obtained.
\end{proofof}

To prove Lemma \ref{complement-upper-bound}, nonetheless, we wish to use an algorithmic technique known as \emph{inductive counting}. This intriguing technique was discovered independently by Immerman \cite{Imm88} and Szelepcs\'{e}nyi \cite{Sze88} in order to prove that $\nl=\co\nl$.

For the description of the proof of  Lemma \ref{complement-upper-bound}, we introduce the following notion of (surface) configurations for a family $\{M_n\}_{n\in\nat}$ of the $n$th $k$-counter 2ncta $M_n$ running in time polynomial, say, $r(n,|x|)$.
A \emph{(surface) configuration} of $M_n$ on input $x$ is of the form $(q,l,\vec{m})$ with $q\in Q$,  $l\in[0,|x|+1]_{\integer}$, and $\vec{m}=(m_1,m_2,\ldots,m_k) \in \{1,\bot\}^k$. This form indicates that $M_n$ is in inner state $q$, scanning the $l$th tape cell, and the $M_n$'s $i$th counter holds a symbol $m_i$ in the top cell for each index $i\in[k]$.
We abbreviate as $CONF_{n,x}$ the configuration space $Q_n\times [0,|x|+1]_{\integer}\times \{1,\bot\}^k)$.
Given two configurations $c_1$ and $c_2$ of $M_n$ on $x$, the notation $c_1\vdash_{x} c_2$ means that $c_2$ is ``reachable'' from $c_1$ by making a single move of $M_n$ on $x$. We abbreviate as $c_1\vdash_x^{t-1}c_t$ a chain of transitions $c_1\vdash_x c_2\vdash_x \cdots \vdash_x c_t$.
Moreover, $\vdash^*_x$ denotes the transitive closure of $\vdash_x$.


\vs{-2}
\begin{proofof}{Lemma \ref{complement-upper-bound}}
Let $k\in\nat^{+}$ and let $\LL=\{(L_n^{(+)},L_n^{(-)})\}_{n\in\nat}$ be any family of promise problems in $\co\twonct_k$.
Since $\co\LL\in\twonct_k$, there is a nonuniform family $\{M_n\}_{n\in\nat}$ of polynomial-size $k$-counter 2ncta's that solves $\co\LL$.
Let $r$ denote a function whose value $r(n,|x|)$ upper-bounds the runtime of $M_n$ on input $x$.
Clearly, each counter holds only a number between $0$ and $r(n,|x|)$.

Our goal is to build, for each index $n\in\nat$, a $(5k+12)$-counter 2ncta $P_n$, which  solves  the promise problem $(L_n^{(+)},L_n^{(-)})$.
A basic idea of constructing such a machine $P_n$ is to provide a procedure of nondeterministically deciding whether $M_n$ rejects input $x$; in other words, whether all computation paths of $M_n$ on $x$ end with only  non-accepting inner states.
For this purpose, we need to ``count'' the number of rejecting computation paths of $M_n$ on $x$ because, if this number matches the total number of computation paths, then we are sure that $M_n$ rejects $x$.
From this follows $\LL\in\twonct_{5k+12}$.

We arbitrarily fix a number $n$ and a valid input $x$.
In what follows, we deal with (surface) configurations of the form $(q,l,\vec{m})$ in $CONF_{n,x}$.
It is important to note that, with the use of additional $k+1$ counters, we can ``enumerate'' all elements in $CONF_{n,x}$, ensuring a linear order on $CONF_{n,x}$.
This fact makes it possible for us to select the elements of $CONF_{n,x}$ sequentially one by one in the following construction of $P_n$.
For each number $i\in[0,r(n,|x|)]_{\integer}$, we define $V_i = \{(q,l,\vec{m})\in CONF_{n,x} \mid (q_0,0,\vec{\bot})\vdash_x^i (q,l,\vec{m}) \}$ and set $N_i=|V_i|$, where $\vec{\bot}=(\bot,\bot,\ldots,\bot)$.
Since $0\leq N_i\leq |CONF_{n,x}|=|Q_n| (r(n,|x|)+1)^k (|x|+2)$,
we wish to calculate the value $N_{r(n,|x|)}$ by inductively calculating each value $N_i$ for $i\in[0,r(n,|x|)]_{\integer}$ using additional counters.

(1)
Hereafter, we intend to calculate $N_i$ inductively in the following fashion. Let $i$ denote an arbitrary number in $[0,r(n,|x|)]_{\integer}$. We need another counter, say, CT1 to remember this value $i$ and remember the value $N_i$ using an additional counter, say, CT2.
When $i=0$, $N_0$ clearly equals $1$. Assume that $1\leq i\leq r(n,|x|)$.
We use two parameters $c$ and $d$, ranging over  $[0,|CONF_{n,x}|]_{\integer}$, whose values are stored into two extra counters, say, CT3 and CT4. Furthermore, we need to remember the current location of $M_n$'s tape head using another counter, say, CT5. To hold $\vec{m}$, we only need extra $k$ counters, called eCT1--eCT$k$.
Most importantly, during the following inductive procedure, we must empty the additionally introduced counters after each round and reuse them to avoid a continuous introduction of new  counters.

(2)
Initially, we set $c=0$ in CT3.
In a sequential way described above, we pick the elements $(q,l,\vec{m})$ from $CONF_{n,x}$ one by one. We store each picked element $(q,l,\vec{m})$ into CT5 and eCT1--eCT$k$, where $q$ is remembered in the form of inner states. For each element $(q,l,\vec{m})$, we nondeterministically select either the process (a) or the process (b) described below in (3), and execute it.
After all elements $(q,l,\vec{m})$ are selected sequentially and either (a) or (b) is executed properly, we define $\hat{N}_i$ to be the current value of $c$ (by moving the content of CT3 into CT2 to empty CT3).

(3)
To utilize the stored values of $i$, $\hat{N}_i$, $l$, and $\vec{m}$, however, we need to copy them into extra $k+3$ counters, say, CT1$'$, CT2$'$, CT5$'$, and eCT1$'$--eCT$k'$ (by bypassing another extra counter, say, CT10 to transfer counter contents) and use these copied counters in the following process.

(a) Choose nondeterministically a computation path, say, $\gamma$ of $M_n$ on $x$ and check whether  $(q_0,0,\vec{0})\vdash_x^i (q,l,\vec{m})$ is true on  this path $\gamma$. This is done by first returning a tape head to the start cell (i.e., the leftmost tape cell), preparing additional $k+1$ counters, say, CT5$''$ and eCT1$''$--eCT$k''$, initializing them (by emptying them), and then simulating $M_n$ on $x$ for the first $i$ steps using CT1$'$ and those new counters.
Let $(p,h,\vec{s})$ denote the configuration reached after $i$ steps from $(q_0,0,\vec{0})$. We then compare between $(p,h,\vec{m})$ and $(q,l,\vec{m})$ by simultaneously decrementing the corresponding counters CT5$'$, eCT1$'$--eCT$k'$, CT5$''$, and eCT1$''$--eCT$k''$.
If the comparison is unsuccessful, then we reject $x$ and halt. Otherwise, we increment the value of $c$ by one.

(b) Initially, we set $d=0$ in CT4. In the aforementioned sequential way, we pick the elements $(p,h,\vec{s})$ from $CONF_{n,x}$ one by one.
For each selected element $(p,h,\vec{s})$, we need to store $(p,h,\vec{s})$ in other $k+1$ counters and copy their contents into the existing $k+1$ counters (as in (1)). For the simulation of $M_n$, to avoid introducing more counters, we reuse CT5$''$ and eCT1$''$--eCT$k''$. Follow nondeterministically a computation path, say, $\xi$ and check whether $(q_0,0,\vec{0})\vdash_x^{i-1} (p,h,\vec{s})$ is true on $\xi$.
As in (1), this is done by the use of the counters. If this is true, then we increment $d$ by one. We then check whether $(p,h,\vec{s})\vdash_x(q,l,\vec{m})$.  If so, reject $x$ and halt. After all elements $(p,h,\vec{s})$ are properly processed without halting, we check whether $d$ matches $\hat{N}_{i-1}$. If not, reject $x$ and halt. Note that the value $c$ does not change.

(4)
Assume that the above inductive procedure ends without entering rejecting states after the parameter $i$ reaches $r(n,|x|)$ and $\hat{N}_{r(n,|x|)}$ is determined. We remember it in a new counter, say, CT3$'$.
We sequentially pick all the elements $(z,t,\vec{e})$ from $(Q_n-Q_{acc,n})\times[0,|x|+1]_{\integer}\times [0,r(n,|x|)]_{\integer}^k$ one by one, store each picked one into CT5 and eCT1--eCT$k$, and conduct the same procedure as described above, except for any execution of (b) and the following point. At the end of the procedure, we obtain the value $c$ in CT3, we check whether this value $c$ equals $\hat{N}_{r(n,|x|)}$ stored in  CT3$'$.
If so, then we accept $x$; otherwise, we reject $x$.

It is possible to prove that, in a certain computation path of $P_n$, the value $\hat{N}_{i}$ correctly represents $N_i$ for any index  $i\in[0,r(n,|x|)]_{\integer}$.
The correctness of the value of $\hat{N}_i$ (i.e., $\hat{N}_{i}=N_{i}$) can be proven by induction on $i$ as in \cite{Imm88}. Therefore, $P_n$ correctly solves $(L_n^{(+)},L_n^{(-)})$.
The above procedure can be implemented on $P_n$ with the total of $5(k+1)+8$ counters. Thus, $\LL$ belongs to $\twonct_{5k+13}$.

This completes the proof of Lemma \ref{complement-upper-bound}.
\end{proofof}

\subsection{Reachability by Multi-Counter Pushdown Automata Families}

The technique of inductive counting was further elaborated by Borodin, Cook, Dymond, Ruzzo, and Tompa \cite{BCD+89} in order to prove that $\logcfl$ is closed under complementation.
Their proof argument, however, uses a Boolean-circuit simulation technique  on the semi-unbounded circuit model, which is well-known to exactly characterize languages in $\logcfl$ \cite{Ven91}.

In what follows, we intend to prove the following inclusion relationship.

\begin{theorem}\label{co-simulation}
For any $k\geq3$, $\co\twonpdct_k\subseteq \twonpdct_3$ holds, and thus $\twonpdct_3=\co\twonpdct_3$.
\end{theorem}

Similarly to Theorem \ref{counter-four}, this theorem is an immediate consequence of the following key lemma together with Proposition \ref{reducing-counter}.

\begin{lemma}\label{pushdown-counter}
For any constant $k\in\nat^{+}$, $\co\twonpdct_k \subseteq \twonpdct_{10k+48}$.
\end{lemma}

Let us describe the proof of Lemma \ref{pushdown-counter}. To specify the number of necessary counters,
the Boolean-circuit simulation technique of \cite{BCD+89} does not seem to be directly applicable to the case of $k$-counter 2npdcta's.
Hence, we need to
seek out a direct simulation of a given complementary 2npdcta for the construction of the desired machine.
In part, we intend to use an argument of Ruzzo \cite[Theorem 1]{Ruz80}.


Now, we return to the proof of Lemma \ref{pushdown-counter}.
Consider a family $\MM=\{M_n\}_{n\in\nat}$ of polynomial-size $k$-counter 2npdcta's $M_n$ with $Q_n$ and $\Gamma_n$.
We assume without loss of generality that $M_n$ halts with the empty stack (containing only $\bot$) and its input-tape head at cell $0$.
Moreover, it is possible to modify $M_n$ so that $M_n$'s push size is exactly $1$ by expanding $\Gamma_n$ appropriately.
Assume also that, during the computation (except for the final step), the stack never becomes empty. This can be enforced by introducing a dummy bottom marker. Here, we call such 2npdcta's \emph{slim}.

Similarly to the case of 2nfcta's in Section \ref{sec:reachability}, given a $k$-counter 2npdcta $M_n$ with $Q_n$ and $\Gamma_n$, a \emph{(surface) configuration} $C$ of $M_n$ on input $x$ is of the form $(q,l,a,\vec{m})$, where $q\in Q_n$, $l\in[0,|x|+1]_{\integer}$, $a\in\Gamma_n$, and $\vec{m}=(m_1,m_2,\ldots,m_k)\in[0,t_x]_{\integer}^k$, where $t_x$ denotes the runtime of $M_n$ on input $x$.
In particular, the initial configuration $C_0$ and the final configuration $C_{fin}$ are of the form   $C_0=(q_0,0,\vec{0})$ and $C_{t_x}=(q_{acc,n},0,\vec{0})$, where $\vec{0}=(0,0,\ldots,0)$.
Let $CONF_{n,x}$ denote the set of all possible configurations of $M_n$ on $x$. Since $CONF_{n,x}=Q_n\times [0,|x|+1]_{\integer}\times \Gamma_n \times [0,t_x]_{\integer}^k$, it follows that $|CONF_{n,x}| \leq |Q_n|(|x|+1)|\Gamma_n|(t_x+1)^k = (n|x|t_x)^{O(1)}$.

\vs{-2}
\begin{proofof}{Lemma \ref{pushdown-counter}}
Let $\LL=\{(L_n^{(+)},L_n^{(-)})\}_{n\in\nat}$ denote any family of promise problems in $\co\twonpdct_k$. Its complement $\co\LL=\{(L_n^{(-)},L_n^{(+)})\}_{n\in\nat}$ thus belongs to $\twonpdct_k$. We then take a family $\{M_n\}_{n\in\nat}$ of polynomial-size $k$-counter 2npdcta's $M_n$ of the form $(Q_n,\Sigma,\{\rhd,\lhd\}, \Gamma_n, k,\{1,\bot\}, q_{0,n}, \bot, Q_{acc,n},Q_{rej,n})$ solving $\co\LL$.
For simplicity, we assume that each machine $M_n$ is slim and that $Q_{acc,n}=\{q_{acc}\}$ and $Q_{rej,n}=\{q_{rej}\}$.

Given an input $x$, $t_x$ denotes the runtime of $M_n$ on $x$.
We respectively denote by $C_0$ and $C_{fin}$ the initial configuration at time $0$ and the accepting configuration at time $t_x$.

Our goal is to obtain the intended $(10k+48)$-counter 2npdcta $N_n$ that simulates the complementary machine of $M_n$.
We first need to modify $M_n$ to another 2npdcta of ``clean form'' to make surface configurations useful in a later simulation of the complementary machine of $M_n$.
Let $x$ be any input.

To obtain such a form, it is useful to consider the special notion of \emph{configuration intervals} (or \emph{conf-intervals}, for short) $\eta$ of the form $(C,s,C',l,r)$ for $C,C'\in CONF_{n,x}$ and $l,r,s\in\nat$. This  quintuple indicates that, at round $r$, $C'$ is reachable from $C$ and they have \emph{distance} $l$ (namely, $C\vdash_x^{l} C'$) and there are $l_1,l_2,\ldots,l_s$ and $C_1,C_2,\ldots,C_s$ whose stack heights are $r$ such that $C\vdash_x C_1$, $C_i\vdash_x C_{i+1}$ for any $i\in[s-1]$, and $C_s\vdash_x C'$.
The checking of a single transition $C\vdash_x C'$ can be done by recovering $C$ from the stack, including an inner state, $k$ counter contents, and a topmost stack symbol. Before recovering $C$, in practice, we need to place a distinguished separator to mark a new bottom marker and then apply a transition function to generate $C''$ from $C$ and compare between $C''$ and the stored $C'$ in the stack. The last entry of $\eta$ is briefly called an ``$r$''-value of $\eta$.

We further introduce the notations concerning conf-intervals. We denote by  $A^{(1)}_{n,x}$ the set of all possible conf-intervals  $\eta$ induced by $M_n$ on $x$; that is, $A^{(1)}_{n,x} = CONF_{n,x}\times [0,t_x-1]_{\integer} \times CONF_{n,x}\times [0,t_x-1]_{\integer}^2$.
A conf-interval $(C,s,C',l,r)$ is succinctly called \emph{basic} if  $s=0$. We also define $A^{(2)}_{n,x}$ to be $A^{(1)}_{n,x}\times A^{(1)}_{n,x}$.

By fixing an appropriate linear ordering on $A^{(1)}_{n,x}$, it is possible to design a procedure, say, $\AAA_{n,x}$ that generates all elements of $A^{(1)}_{n,x}$ sequentially according to this linear order. This procedure $\AAA_{n,x}$ helps us find the $i$th element of $A^{(1)}_{n,x}$ using two extra counters when $i$ is provided in another counter. We note that $\AAA_{n,x}$ can be implemented on an appropriate $7$-counter 2nfcta.

To handle conf-intervals in the following modification process, we view each conf-interval as an  $O(\log{n|x|t_x})$-length string and allow the stack to hold such a string as a ``single'' entry of the stack although it actually occupies a block of $O(\log{n|x|t_x})$ consecutive stack cells.
Moreover, by running $\AAA_{n,x}$ twice, we can find a pair of the $i$th element and the $j$th element of $A^{(1)}_{n,x}$ when $i,j$ are provided in two counters.


\begin{figure}[t]
\centering
\includegraphics*[height=3.5cm]{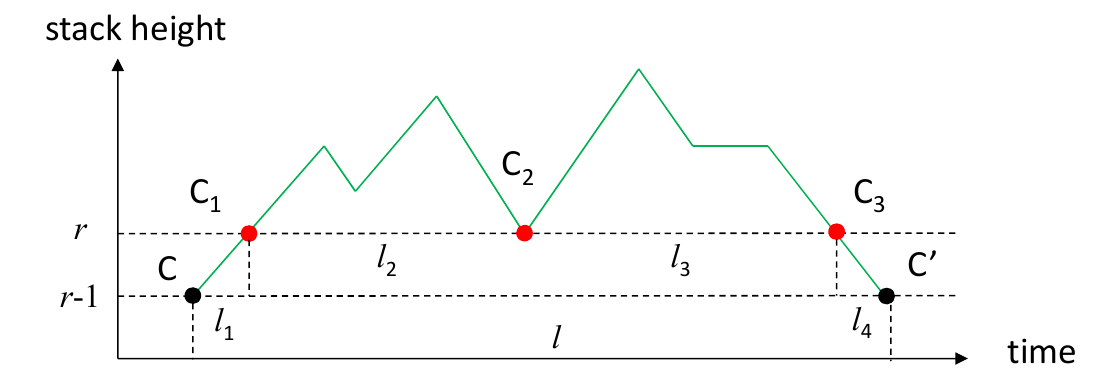}
\caption{A stack history of a 2npdcta. Rounds $r-1$ and $r$ loosely correspond to stack height. From a conf-interval $(C,0,C',l,r-1)$, we obtain another conf-interval $(C_1,1,C_3,l_2+l_3,r)$ with $l=l_1+l_2+l_3+l_4$ and $l_1=l_4=1$, from which we further obtain two more conf-intervals $(C_1,0,C_2,l_2,r)$ and $(C_2,0,C_3,l_3,r)$ by applying (i) and (ii) in the nondeterministic phase (2) in the proof of Lemma \ref{pushdown-counter}.
}\label{fig:stack-height}
\end{figure}


\ms

[I]
In this first stage, we intend to modify $M_n$ into another $(5k+8)$-counter 2npdcta, say, $\tilde{M}_n$ so that the obtained family $\{\tilde{M}_n\}_{n\in\nat}$ still solves $\co\LL$.
The modification of $M_n$ into $\tilde{M}_n$ is carried out by taking the following recursive steps (1)--(4).

(1) [initial phase] In this first phase, we generate the conf-interval  $\eta_0= (C_0,0,C_{fin},t_x,0)$ and store it in $\tilde{M}_n$'s stack. We then move to the phase (3).

(2) [accepting phase]
If the stack is empty, then we accept the input and terminate the procedure. Otherwise, we must move to the phase (3).
Note that basic conf-intervals $(\eta_1,\eta_2,\ldots,\eta_d)$ having the same $r$-value for $d\in\nat^{+}$ are stored consecutively from the top in the stack.
In the later construction of $N_n$, this phase will be converted into a specific nondeterministic phase.

(3) [nondeterministic phase]
Assume that the topmost stack entry $\eta$ is a conf-interval of the form $(C,s,C',l,r-1)$. We pop it and copy its elements into $2k+3$ separate counters.
See Fig.~\ref{fig:stack-height} for an illustration of several conf-intervals obtainable in the cases below.

(i) In the case of $l\geq2$ and $s=0$, we nondeterministically generate two configurations $C_{1},C_{2}\in CONF_{n,x}$ and two numbers  $l_1\in[0,t_x]_{\integer}$ and $s_1\in[0,t_x-1]_{\integer}$ to form one conf-interval $\eta_1= (C_1,s_1,C_2,l_1,r)$. We then check whether the condition
$CND^{(1)}_{n,x,\eta}(\eta_1)$ is true, where $CND^{(1)}_{n,x,\eta}(\eta_1)$ is true if $C\vdash C_1$, $C_2\vdash C'$, $l_1=l-2$ hold, and   $CND^{(1)}_{n,x,\eta}(\eta_1)$ is false otherwise. If not, we immediately reject the input and halt. Otherwise, we store $\eta_1$ in the stack and then move back to the phase (3).

(ii) In the case of $l\geq2$ and $s\geq1$, similarly to (i), we nondeterministically generate three configurations $C_{1},C_{2},C_{3} \in CONF_{n,x}$ and four numbers $l_1,l_2\in[0,t_x]_{\integer}$ and $s_1,s_2\in[0,t_x-1]_{\integer}$ to form two conf-intervals  $\eta_1= (C_1,s_1,C_2,l_1,r-1)$ and $\eta_2= (C_2,s_2,C_3,l_2,r-1)$.
These generated quintuples are stored in $3k+4$ separate counters.
We then check whether $CND^{(2)}_{n,x,\eta}(\eta_1,\eta_2)$ is true using an extra counter, where $CND^{(2)}_{n,x,\eta}(\eta_1,\eta_2)$ is true if $C_1=C$, $C_3=C'$, $s_1=0$, $s_2=s-1$, and $l=l_1+l_2$ hold, and    $CND^{(2)}_{n,x,\eta}(\eta_1,\eta_2)$ is false otherwise.
If not, we reject and halt. Otherwise, we store $\eta_1$ and $\eta_2$ in this order into the stack and then move back to the phase (3).

(iii) In the case of $l=0$ and $s=0$, we must go to the phase (4).

(4) [deletion phase]
Assume that $\eta= (C,s,C',l,r)$ is in the topmost stack cell.
We pop it and store its elements into $2k+3$ counters. In the case of $l=s=0$, if $C\neq C'$, then we reject and halt; otherwise, we delete both $\eta$ and the next stored conf-interval of the form $(\tilde{C},s',\tilde{C}',2,r)$ for certain $\tilde{C}$ and $\tilde{C}'$. We then move to the phase (2).

\ms

[II]
In this second stage, we consider the ``complementary'' machine, say, $\overline{M_n}$ obtained from $\tilde{M}_n$ by the following modification. We swap between nondeterministic moves and universal moves and also between $Q_{acc,n}$ and $Q_{rej,n}$. It thus follows that all computation paths of $\overline{M_n}$ are ``accepting'' for each input $x\in L_n^{(+)}$ and that there exists a ``rejecting'' computation path of $\overline{M_n}$ for each  $x\in L_n^{(-)}$. This clearly implies that the obtained family $\{\overline{M_n}\}_{n\in\nat}$ can solve $\LL$.

\ms

[III]
Recall that $\tilde{M}_n$ uses $(5k+8)$-counters. For convenience, we set $k'=5k+8$. Now, we construct the desired $(5k'+8)$-counter 2npdcta's $N_n$ that simulates the complementary  machine $\overline{M_n}$.

(1$'$) This initial phase is in essence similar to (1), including the same accepting configuration $C_{fin}$ of $\tilde{M}_n$. We then move to the phase (3$'$).

(2$'$) Next, we modify the phase (2) of $\tilde{M}_n$. If the stack is empty, then we accept the input and halt.
Assuming otherwise, we move to the phase (3$'$).

(3$'$) This is a key phase of $N_n$.
Let $\eta=(C,s,C',l,r-1)$ denote a topmost entry of $N_n$'s stack.
We pop it and keep it using $(2k'+3)$-counters.
Since $\tilde{M}_n$ makes nondeterministic choices of elements in $A^{(1)}_{n,x}$ in (i),  $\overline{M_n}$ makes universal choices of
conf-intervals in the phase corresponding to (i).
However, since $N_n$ is not allowed to make universal moves, we instead make the following nondeterministic moves (i$'$)--(iii$'$).

(i$'$) Assume that $l\geq2$ and $s=0$.
We wish to turn $\overline{M_n}$'s universal moves to stack operations of  sequentially storing in the stack $\overline{M_n}$'s nondeterministically chosen conf-intervals. For this purpose, we run $\AAA_{n,x}$ to sequentially generate all elements $\eta_1$ in $A^{(1)}_{n,x}$ and check whether $CND^{(1)}_{n,x,\eta}(\eta_1)$ is true. If not, we skip $\eta_1$ or else we push $\eta_1$ into the stack.
Since $|A^{(1)}_{n,x,\eta}|=(n|x|t_x)^{O(1)}$, it is possible to store all elements $\eta_1$ of $A^{(1)}_{n,x,\eta}$ satisfying $CND^{(1)}_{n,x,\eta}(\eta_1)$.
This entire process can be implemented on a 2npdcta by incrementing a counter holding $i$ of $\AAA_{n,x}$.
We thus need $(3k'+4)$-counters to execute this whole process.
After all elements of $A^{(1)}_{n,x,\eta}$ are properly processed, we move back to the phase (3$'$).

(ii$'$) Assume that $l\geq2$ and $s\geq1$. Notice that all possible nondeterministic choices of $\tilde{M}_n$ come from the set $A^{(2)}_{n,x}$. Similarly to (i$'$), we sequentially generate all pairs $(i,j)$ with $i,j\in[0,|A^{(1)}_{n,x,\eta}|]_{\integer}$ and $i<j$. For each chosen pair $(i,j)$, we produce the $i$th element $\eta_1$ and the $j$th element $\eta_2$ by running $\AAA_{n,x}$ twice. Since we do not need to check both $\eta_1$ and $\eta_2$, we nondeterministically choose either $\eta_1$ or $\eta_2$ and store only the chosen one in the stack. After all pairs are processed, we move back to the phase (3$'$).

(iii$'$) If $l=0$ and $s=0$, then we must go to the phase (4$'$).

(4$'$) Given a topmost entry $\eta= (C,s,C',l,r)$ in the stack, in the case of $l=0$ and $s=0$, we reject and halt if $C=C'$; otherwise, we delete $\eta$ and the next stack entry of the form $(\tilde{C},\tilde{s},\tilde{C}',2,r)$. We then move to the phase (3$'$).

Overall, $N_n$ requires $(5k'+8)$-counters, as requested.
\end{proofof}

\section{Relationships to Parameterized Decision Problems}\label{sec:parameterization}

Kapoutsis \cite{Kap09,Kap12} introduced the succinct notation of ``$/\poly$'' for $\twod$, $\twon$, etc. to express the polynomial-ceiling restriction of all valid inputs. He then demonstrated the following equivalence: $\twon/\poly \subseteq \twod$ iff $\nl\subseteq \dl/\poly$ \cite{Kap14}. This equivalence establishes a bridge to the standard complexity classes $\dl$ and $\nl$.
A key argument of his proof is pivoted around an $\nl$-complete problem called ``two-way liveness'' and its variant  \cite{Kap14,SS78} together with the use of ``homomorphic reductions.''
Yamakami \cite{Yam22a} later presented another proof technique, which is not relying on any particular ``complete'' problems or any particular ``reductions''.
This ``generic'' proof technique exploits a close connection between parameterized decision problems and families of promise problems solvable by certain finite and pushdown automata families. This connection was already used to establish bridges between the collapses of nonuniform {(stack-)state} complexity classes and those of standard advised complexity classes \cite{Yam21,Yam22a,Yam22b}.

In what follows, we quickly review necessary terminology, introduced in \cite{Yam23a} and applied in \cite{Yam21,Yam22a,Yam22b} to nonuniform {(stack-)state} complexity issues.
A \emph{parameterized decision problem} over alphabet $\Sigma$ is  of the form $(L,m)$, where $L\subseteq \Sigma^*$ and $m(\cdot)$ is a size parameter (i.e., a mapping of $\Sigma^*$ to $\nat$).
Any size parameter computable by an appropriate log-space deterministic Turing machine (DTM) is called a \emph{logspace size parameter}. A typical example is $m_{\parallel}$ defined as $m_{\parallel}(x)=|x|$ for all $x\in\Sigma^*$.
All parameterized decision problems whose size parameters $m$ is \emph{polynomially honest} (i.e., $|x|\leq q(m(x))$ for all $x\in\Sigma^*$ for an absolute polynomial $q$) form the complexity class $\phsp$.
Refer to \cite{Yam22a,Yam23a} for more information.

Let us recall that $\logcfl$ is characterized in terms of polynomial-time log-space aux-2npda's \cite{Sud78}.
We use the notation $\para\mathrm{NAuxPDA,\!\!TISP}(n^{O(1)},\log{n})/\poly$ to denote the collection of all parameterized decision problems $(L,m)$ with logspace size parameters $m$ solvable by advised aux-2npda's running in time $m(x)^{O(1)}$ and space $O(\log{m(x)})$ with the use of advice strings of length polynomial in $m(x)$, where $x$ represents an ``arbitrary'' input.
As a special case, if $m$ is fixed to $m_{\parallel}$, since $(L,m_{\|})$ can be viewed as just $L$,
we obtain the non-parameterized class $\logcfl/\poly$. In a similar fashion, we can define $\para\mathrm{UAuxPDA,\!\!TISP}(n^{O(1)},\log{n})/\poly$ and $\para\co\mathrm{NAuxPDA,\!\!TISP}(n^{O(1)},\log{n})/\poly$ using underlying advised aux-2upda's and ``complementary'' advised aux-2npda's, respectively. 
Following Section \ref{sec:advised-machine}, we also abbreviate the above-mentioned parameterized advised complexity classes as $\para\logcfl/\poly$, $\para\logucfl/\poly$, and $\para\co\logcfl/\poly$.

Hereafter, we denote by $(P,m)$ any parameterized decision problem and by  $\KK=\{(K_n^{(+)},K_n^{(-)})\}_{n\in\nat}$ and $\PP=\{(P_n^{(+)},P_n^{(-)})\}_{n\in\nat}$ two arbitrary families of promise problems over a common alphabet $\Sigma$, which does not include a special symbol $\#$ used as a separator. For convenience, we set $\Sigma_{\#}=\Sigma\cup\{\#\}$ and $S_{\#,n}^{(\geq2)} = \{w\in (\Sigma_{\#})^{n}\mid \text{ $w$ contains more than one $\#$ }\}$ for each number $n\in\nat$.
The \emph{size-incorporated family} $\KK$ of $\PP$ is defined by setting $K_n^{(+)} = \{1^n\# x\mid x\in P_n^{(+)}\}$ and $K_n^{(-)}=\{1^n\# x\mid x\notin P_n^{(+)}\} \cup S_n$ for any $n\in\nat$,  where $S_n$ denotes the set $\{z\# x\mid z\in \Sigma^n-\{1^n\},x\in \Sigma^*\}\cup \Sigma^n \cup S_{\#,n}^{(\geq 2)}$.

\begin{definition}\label{def-induce}
(1) Given a parameterized decision problem $(L,m)$ over $\Sigma$, we say that $\PP$ is \emph{induced from} $(L,m)$ if, for any index $n\in\nat$, $P_n^{(+)}= L\cap \Sigma_{(n)}$ and $P_n^{(-)}=\overline{L}\cap \Sigma_{(n)}$, where $\overline{L}=\Sigma^*-L$ and $\Sigma_{(n)}=\{x\in\Sigma^*\mid m(x)=n\}$.

(2) Given a family $\PP$ of promise problems, we say that $(K,m)$ with $K\subseteq (\Sigma_{\#})^*$ and $m:(\Sigma_{\#})^*\to\nat$ is \emph{induced from} $\PP$ if there exists a size-incorporated family  $\KK=\{(K_n^{(+)},K_n^{(-)})\}_{n\in\nat}$ of $\PP$ for which
(i) $K=\bigcup_{n\in\nat} K_n^{(+)}$, (ii) $m(w)=n$ holds for any string $w \in \{1^n\# x \mid x\in P_n^{(+)}\cup P_n^{(-)}\}$, and (iii)  $m(w)=|w|$ holds for all other strings $w$.
\end{definition}

In the above definition, the set $\bigcup_{n\in\nat} K_n^{(-)}$ equals $(\bigcup_{n\in\nat}\{1^n\# x\mid x\notin P_n^{(+)}\})\cup (\bigcup_{n\in\nat}S_n) = \{1^n\# x\mid x\notin P_n^{(+)}\}\cup \Sigma^* \cup \{w\mid \text{ $w$ contains more than one $\#$ }\}$. Thus,
$\overline{K}$ coincides with $\bigcup_{n\in\nat} K_n^{(-)}$.
See \cite{Yam22a} for more information.

A family $\KK$ is said to be an \emph{extension} of $\PP$ if  $P_n^{(+)}\subseteq K_n^{(+)}$ and $P_n^{(-)}\subseteq K_n^{(-)}$ hold for any $n\in\nat$.
As a quick example, if a 2nfa $M_n$ solves $(P_n^{(+)},P_n^{(-)})$, then the pair $(L(M_n),L(\overline{M_n}))$ is an extension of $(P_n^{(+)},P_n^{(-)})$ because $M_n$ accepts all inputs in $P_n^{(+)}$ and rejects all inputs in $P_n^{(-)}$, where $\overline{M_n}$ expresses a ``complementary'' machine\footnote{A complementary machine is a machine whose final decision is the opposite of the original machine.} of $M_n$.
Moreover, we say that the family $\PP$ is \emph{L-good} if the set $\{1^n\# x \in (\Sigma_{\#})^* \mid n\in\nat, x\in P_n^{(+)}\cup P_n^{(-)}\}$ belongs to $\dl$ (log-space complexity class).
An \emph{L-good extension} of $\PP$ means an $\dl$-good family that is also an extension of $\PP$. We further expand the notion of $\dl$-goodness to a complexity class.
A collection $\FF$ of promise problem families is \emph{L-good} if, for each family $\PP$ of $\FF$, there is an $\dl$-good extension of $\PP$ in $\FF$.

The $\dl$-goodness can be proven for, e.g., $\twodpd$ and $\twonpd$ in a way similar to \cite{Yam22a}.
The $\dl$-goodness of $\co\twonpd$ also follows.

\begin{lemma}
The following complexity classes are all $\dl$-good: $\twod$, $\twon$, $\twodpd$, $\twonpd$, and $\co\twonpd$. Even if the runtime restriction is imposed on underlying machines, they remain $\dl$-good.
\end{lemma}

\begin{proof}
We first show the $\dl$-goodness of $\twonpd$.
Let us take an arbitrary family $\LL=\{(L_n^{(+)},L_n^{(-)})\}_{n\in\nat}$ of promise problems over alphabet $\Sigma$ in $\twonpd$.
There exists a family $\MM=\{M_n\}_{n\in\nat}$ of polynomial-size 2npda's that solves $\LL$.
Let $L_{(n)}= L_n^{(+)}\cup L_n^{(-)}$ for any $n\in\nat$.
Although $M_n$ is not even required to halt on inputs outside of $L_{(n)}$, we can easily modify $M_n$ to halt on all of its computation paths on any input.
We then define $K_n^{(+)} = \{x\in\Sigma^*\mid \text{ $M_n$ accepts $x$ }\}$ and $K_n^{(-)} = \{x\in\Sigma^*\mid \text{ $M_n$ rejects $x$ }\}$. We then set $\KK =\{(K_n^{(+)},K_n^{(-)})\}_{n\in\nat}$.
Obviously, $\KK$ belongs to $\twonpd$.

Here, we claim that $\KK$ is $\dl$-good.
Since $K_n^{(+)}\cup K_n^{(-)} = \Sigma^*$, $L_n^{(+)}\subseteq K_n^{(+)}$, and  $L_n^{(-)}\subseteq K_n^{(-)}$, it then follows that $\KK$ is an extension of $\LL$. We then set $A=\{1^n\# x\mid n\in\nat, x\in K_n^{(+)}\cup K_n^{(-)}\}$, which equals $\{1^n\# x\mid n\in\nat, x\in\Sigma^*\}$. Clearly, this set $A$ is recognized deterministically using only logarithmic space, and thus $A$ belongs to $\dl$.

In a similar way, we can prove the $\dl$-goodness of $\twon$. Furthermore, the $\dl$-goodness of $\co\twonpd$ comes from the fact that the set $\{1^n\#x\mid x\in \nat, x\in P_n^{(+)}\cup P_n^{(-)}\}$ does not alter by exchanging between $P_n^{(+)}$ and $P_n^{(-)}$.
\end{proof}


Here, we introduce the important notion of canonical twins.

\begin{definition}\label{def-twin}
A pair $(\CC,\DD)$ of complexity classes is said to be \emph{canonical twins} if the following three statements (a)--(c) are all satisfied. If only the two statements (a)--(b) are satisfied, then $(\CC,\DD)$ is called \emph{pseudo-canonical twins}.
Let $L$ and $K$ be any two decision problems over the common alphabet $\Sigma$ and let $m$ denote any logspace size parameter over $\Sigma$. Let $\LL$ be any family of promise problems.
\renewcommand{\labelitemi}{$\circ$}
\begin{enumerate}\vs{-1}
  \setlength{\topsep}{-2mm}%
  \setlength{\itemsep}{1mm}%
  \setlength{\parskip}{0cm}%

\item[(a)] If $\LL$ is induced from $(L,m)$, then $(L,m)\in\para\CC/\poly\cap \PHSP$ iff $\LL\in\DD/\poly$.

\item[(b)]  If $(K,m)$ is induced from $\LL$, then $(K,m)\in\para\CC/\poly\cap \PHSP$ implies $\LL\in\DD/\poly$.

\item[(c)]  If $(K,m)$ is induced from $\LL$, then $\LL\in\DD/\poly$ implies $(K,m)\in\para\CC/\poly\cap \PHSP$.
\end{enumerate}\vs{-2}
\end{definition}

A key statement of this work is given as the following proposition. We write $\LGOOD$ to denote the collection of all $\dl$-good families of promise problems.

\begin{proposition}\label{L-good-induced}
(1) Consider the set $\{(\logdcfl,\twodpd), (\logcfl,\twonpd)\}$. 
Every element $(\CC,\DD)$ of this set is canonical twins.

(2) Consider the set $\{(\ul,\twou), (\logucfl,\twoupd)\}$. Every element $(\CC,\DD)$ in this set is pseudo-canonical twins.

(3) Consider the set $\{(\ul,\twou\cap\LGOOD),(\logucfl,\twoupd\cap\LGOOD)\}$. Every $(\CC,\DD)$ is canonical twins.
\end{proposition}

\begin{proof}
(1)
A basic argument was presented in \cite{Yam22a}. Here, we wish to elaborate this argument and present it in a generic form.
We thus demonstrate in detail the validity of the statements (a)--(c) of Definition \ref{def-twin} for the particular case of $\CC=\logcfl$ and $\DD=\twonpd$. The other case of $\CC=\logdcfl$ and $\DD=\twodpd$ is similarly proven.

(a) By elaborating a basic argument of \cite{Yam21,Yam22a,Yam22b}, we prove (a).
Let us assume that $\LL$ is induced from $(L,m)$.
We begin with an assumption that $(L,m)$ is in $\para\logcfl/\poly\cap \PHSP$.
There are a polynomially-bounded advice function $h$ and an advised aux-2npda $M$ that recognizes $K$ with the help of $h$ in time $m(x)^{O(1)}$ and space $O(\log{m(x)})$. Since $m$ is polynomially honest, a suitable polynomial $q$ satisfies $|x|\leq q(m(x))$ for all $x\in\Sigma^*$. We also assume that $|h(|x|)|\leq q(m(x))$ for all $x$.
Remark that $L_n^{(+)} = L\cap \Sigma_{(n)}$ and $L_n^{(-)} = \overline{L}\cap\Sigma_{(n)}$ for any $n\in\nat$, where $\Sigma_{(n)}=\{x\mid m(x)=n\}$. For any string $x$ in $L_n^{(+)}\cup L_n^{(-)}$, since $m(x)=n$, we obtain $|x|\leq q(m(x))=q(n)$.

Hereafter, we wish to derive the membership of $\LL$ to $\twonpd/\poly$.
Fix $n$ arbitrarily and consider the following 2npda $N_n$, as in \cite[Lemma 5.3]{Yam22a}. Let $\hat{h}(n)$ denote the string $h(0)\# h(1)\# \cdots \# h(q(n))$. Notice that $|\hat{h}(n)|\leq \sum_{i=0}^{q(n)} (|h(i)|+1) \leq (q(n)+1)^2$.
\vs{-1}
\begin{quote}
On input $x$, we first retrieve $n$. This is possible since $n$ is a fixed constant for $N_n$. If $|x|>q(n)$, then we reject the input. Otherwise, we retrieve the string $h(|x|)$ from $\hat{h}(n)$. We then run $M$ on $x$ using the advice string $h(|x|)$.
\end{quote}
\vs{-1}
\n The runtime of $N_n$ is upper-bounded by $m(x)^{O(1)}\cdot (q(n)|x|)^{O(1)}$ ($\subseteq (n|x|)^{O(1)}$). The machine $N_n$ requires a simulation of $M$ whose space usage is $O(\log{m(x)}) \subseteq O(\log{n})$.  Therefore, it is possible to encode the work tape content of $M$ together with $\hat{h}(n)$ into $N_n$'s inner states. This makes the total number of inner states needed for $N_n$ be $n^{O(1)}$.
Moreover, $N_n$ accepts all $x\in L_n^{(+)}$ and rejects all $x\in L_n^{(-)}$.
We thus conclude that  $\LL$ is in $\twonpd/\poly$.

Conversely, we assume that $\LL$ is in $\twonpd/\poly$.
Take a family $\{M_n\}_{n\in\nat}$ of polynomial-size 2npda's solving $\LL$ in time $p(n,|x|)$ for a suitable polynomial $p$.  There is a polynomial $q$ satisfying $L_n^{(+)}\cup L_n^{(-)}\subseteq \Sigma^{\leq q(n)}$ and $|Q_n|\leq q(n)$ for all $n\in\nat$, where $Q_n$ is a set of $M_n$'s inner states.
Since $L_n^{(+)} =L\cap \Sigma_{(n)}$ and $L_n^{(-)} = \overline{L}\cap\Sigma_{(n)}$, we obtain $\Sigma_{(n)} = L_n^{(+)}\cup L_n^{(-)} \subseteq \Sigma^{\leq q(n)}$. Hence, it follows that $|x|\leq q(n) = q(m(x))$ for any $x\in \Sigma_{(n)}$. In other words, since $\bigcup_{n\in\nat} \Sigma_{(n)} = \Sigma^*$, $m$ is polynomially honest.

Our goal is to prove that $(L,m)\in \para\logcfl/\poly\cap\PHSP$. We first encode each $M_n$ using an appropriately defined encoding scheme. With the use of such a scheme, $\pair{M_n}$ denotes the encoded string that expresses $M_n$.
We then define $h(n)=\pair{M_n}$ and set $\hat{h}(n)$ to be $h(0)\# h(1)\# \cdots \# h(q(n))$.
As in \cite[Lemma 5.2]{Yam22a}, let us consider the following advised aux-2npda $N$ taking $\hat{h}$ as advice.
\vs{-1}
\begin{quote}
On input $x$, compute $m(x)$ and set $n=m(x)$. Select $h(n)$ from the given advice string $\hat{h}(|x|)$. This is possible because $|x|\leq q(n)$. Recover $M_n$ from $h(n)=\pair{M_n}$ and simulate $M_n$ on $x$ within time $p(n,|x|)$.
\end{quote}
\vs{-1}
\n The runtime of $N$ is $O(p(n,|x|))\cdot (n|x|)^{O(1)}$, which equals  $m(x)^{O(1)}$, because of $n=m(x)$. Note that the working space is $O(\log{|x|})+ O(\log{|Q_n|})$, which equals $O(\log{n})$, where $Q_n$ denotes a set of $M_n$'s inner states.
These bounds imply that $(L,m)$ belongs to $\para\logcfl/\poly$.

(b)
Next, we wish to prove (b) in a way similar to a conversion given in \cite[Lemma 5.5]{Yam22a}. Assume that $(K,m)$ is induced from $\LL$ and that $(K,m)\in\para\logcfl/\poly\cap \PHSP$.
For each $n\in\nat$, let
$S_n = \{z\# x\mid z\in \Sigma^n-\{1^n\}, x\in \Sigma^*\}\cup \Sigma^n \cup S_{\#,n}^{(\geq2)}$.
Since $(K,m)$ is induced from $\LL$, there exists a size-incorporated family $\KK=\{(K_n^{(+)},K_n^{(-)})\}_{n\in\nat}$ of $\LL$ with $K_n^{(+)}=\{1^n\# x\mid x\in L_n^{(+)}\}$ and $K_n^{(-)}=\{1^n\# x\mid x\notin L_n^{(+)}\}\cup S_n$ satisfying that $K$ has the form $\bigcup_{n\in\nat} K_n^{(+)}$.
As noted before, $\overline{K}$ equals $\bigcup_{n\in\nat} K_n^{(-)} = (\bigcup_{n\in\nat} \{1^n\# x\mid x\notin L_n^{(+)}\}) \cup (\bigcup_{n\in\nat}S_n)$.

Since $(K,m)\in\para\logcfl/\poly\cap \PHSP$, there are a polynomially-bounded advice function $h$ and an advised aux-2npda $M$ recognizing $K$ with $h$ in $m(x)^{O(1)}$ time and $O(\log{m(x)})$ space. Moreover, we take a polynomial $q$ satisfying $|h(|x|)|\leq q(m(x))$ for all $x$.
Note that $m(1^n\#x) =n$ for all $x\in L_n^{(+)}\cup L_n^{(-)}$ and $m(w)=|w|$ for all other strings $w$.

We wish to construct a family $\{M_n\}_{n\in\nat}$ of 2npda's. Fix $n$ arbitrarily.
\vs{-1}
\begin{quote}
On input $x$, generate $1^n$. If $|x|>q(n)$, then reject. Otherwise, recover the string $h(1^n\#x)$ from inner states and run $M$ on $1^n\# x$ using $h(1^n\#x)$ as an advice string.
\end{quote}
\vs{-1}
\n For any $x$ of length at most $q(n)$, the runtime of $M_n$ on $x$ is $m(1^n\#x)^{O(1)} \subseteq n^{O(1)}$ because of $m(1^n\#x)=n$. Note that $|h(n+|x|+1)|= |h(|1^n\#x|)|\leq q(m(1^n\#x))=q(n)$.
Since the work space usage of $M$ on $1^n\#x$ is $O(\log{m(1^n\#x)})\subseteq O(\log{n})$, we can encode work tape contents of $M$ and $h(n+|x|+1)$ into inner states. The state complexity of $M_n$ then becomes $n^{O(1)}$. Obviously, $M_n$ accepts (resp., rejects) all $x$ satisfying $1^n\# x\in K_n^{(+)}$ (resp., $1^n\# x\in K_n^{(-)}$). Therefore, $\LL$ belongs to $\twonpd/\poly$.

(c) Assume that $(K,m)$ is induced from $\LL$.  Similarly to (b), we take a size-incorporated family $\KK=\{(K_n^{(+)},K_n^{(-)})\}_{n\in\nat}$ of $\LL$ with $K_n^{(+)}=\{1^n\#x\mid x\in L_n^{(+)}\}$ and $K_n^{(-)}=\{1^n\#x\mid x\notin L_n^{(+)}\}\cup S_n$. Moreover, we take a polynomial $q$ such that $L_n^{(+)}\cup L_n^{(-)}\subseteq \Sigma^{q(n)}$  for all $n\in\nat$ and that a set $Q_n$ of $M_n$'s inner states satisfies $|Q_n|\leq q(n)$ for all $n\in\nat$.
Similarly to (a), we then define $\hat{h}(m) = h(0)\# h(1)\# \cdots \# h(q(m))$, where $h(i)=\pair{M_i}$ for all $i\in[0,q(m)]_{\integer}$. In particular, when $m=|1^n\#x|=n+|x|+1$, since $q$ has only nonnegative coefficients, $q(n)\leq q(m)$ holds.
In addition, since $\twonpd/\poly$ is $\dl$-good, we can assume without loss of generality that $\LL$ is $\dl$-good (otherwise, we can take an $\dl$-good extension of $\LL$ as a new $\LL$). There is a polynomial-time log-space DTM $N_0$ that recognizes $\{1^n\#x\mid n\in\nat, x\in L_n^{(+)}\cup L_n^{(-)}\}$.

Consider the following procedure.
\vs{-1}
\begin{quote}
On input $w$, if $w$ is in $\Sigma^*\cup (\bigcup_{n\in\nat}S_n)$, then we reject. Assume otherwise. Since $S_n=\{z\# x\mid z\in\Sigma^n-\{1\}^*, x\in\Sigma^*\}\cup \Sigma^n\cup S_{\#,n}^{(\geq n)}$, $w$ has the form $1^n\#x$ for certain $n\in\nat$ and $x\in\Sigma^*$.
Since $1^n\# x\in K_n^{(+)}\cup K_n^{(-)}$, $x$ is in either $L_n^{(+)}$ or $\Sigma^n-L_n^{(+)}$. (*) We run $N_0$ on $1^n\#x$ to check whether $x\in L_n^{(+)}\cup L_n^{(-)}$ and, if not, reject immediately. We then recover $M_n$ from the advice string $\hat{h}(n+|x|+1)$ and simulate $M_n$ on $x$.
\end{quote}
\vs{-1}
\n We remark that the instruction (*) of the above procedure is not necessary. We include it here only for the later reference in (3c). It is not difficult to carry out this procedure on an appropriate advised aux-2npda $N$ with the advice function $\hat{h}$.
The runtime of $N$ on $w$ is at most $p(n,|x|)$, which is at most $p(n,q(n)) =n^{O(1)} =m(1^n\#x)^{O(1)}$. The work space is $O(\log{|Q_n|})$, which is at most $O(\log{q(n)}) = O(\log{n}) = O(\log{m(1^n\#x)})$ because of $q(n)=n^{O(1)}$.


(2) Our task is to prove the statements (a)--(b) of Definition \ref{def-twin} for the particular case of $\CC=\logucfl$ and $\DD=\twoupd$. The other case of $\CC=\ul$ and $\DD=\twou$ also follows by a similar argument.

(a) Here, we wish to show (a) by following an argument of (1). Assume that $\LL$ is induced from $(L,m)$. We begin with assuming that  $(L,m)$ falls in $\para\logucfl/\poly\cap \PHSP$. Take a polynomially-bounded advice function $h$ and an advised aux-2upda $M$ that recognizes $L$ in time $m(x)^{O(1)}$ using space $O(\log{m(x)})$ with the help of the advice string $h(m(x))$ given according to each input $x$. Since $L_n^{(+)}=L\cap \Sigma_{(n)}$ and $L_n^{(-)}=\overline{L}\cap\Sigma_{(n)}$ (where $\overline{L}=\Sigma^*-L$ and $\Sigma_{(n)}=\{x\mid m(x)=n\}$), it is possible to convert $M$ with $h$ working on all inputs $x$ in $\Sigma_{(n)}$ into an appropriate 2upda $N_n$ by integrating $h$ into $M$ and then simulating $M$ on the 2upda $N_n$ using polynomially many inner states.
Thus, $\LL$ belongs to $\twoupd$.

On the contrary, assume that $\LL\in\twoupd/\poly$. Take a family $\MM=\{M_n\}_{n\in\nat}$ of polynomial-size 2upda's solving $\LL$ in polynomial time. We can convert $\MM$ into a single polynomial-time log-space advised aux-2upda $N$ solving $(L,m)$ with a polynomially-bounded advice function $h$.

(b) Next, we intend to prove the statement (b).
Assume that $(K,m)$ is induced from $\LL$ via a size-incorporated family $\KK$ of $\LL$.
If $(K,m)$ belongs to $\para\logucfl/\poly\cap \PHSP$, then there are a polynomially-bounded advice function $h$ and an advised aux-2upda $M$ that solves $(K,m)$ in $O(m(z)^k)$ time using $O(\log{m(z)})$ space on every input $z$ of the form $(x,h(m(x)))$, where $k$ is an absolute constant.
Let $n$ denote an arbitrary number in $\nat$. We wish to construct the following 2npda $N_n$.

We first encode a content of the work tape into an inner state and  approximately define a new transition function of $N_n$ so that it simulates $M$ step by step.
Consider arbitrary $n$ and $x\in\Sigma_{(n)}$. Since $m$ and $h$ are polynomially bounded, the space usage of $N_n$ is upper bounded by $O(\log{m(z)})\subseteq O(\log{n})$ and the runtime of $N_n$ is upper-bounded by $O(m(z)^k)\subseteq O(n^{k'})$ for an appropriate constant $k'>k$.
By the definition of $N_n$, the family $\{N_n\}_{n\in\nat}$ turns out to solve $\LL$. We thus conclude that $\LL\in\twoupd/\poly$.

(3)
Consider the case of $\CC=\logucfl$ and $\DD=\twoupd\cap\LGOOD$. The statements (a)--(b) are the same as those of (2). We then focus on (c).
Recall the instruction (*) given in the description of $N$'s procedure in (1c). Unlike the case of $\twoupd/\poly$, this instruction (*) is necessary to make $N$ be an advise aux-2upda. The $\dl$-goodness of $(\twoupd\cap\LGOOD)/\poly$, nonetheless, guarantees the instruction (*).
\end{proof}

\section{Effects of Polynomial Ceiling Bounds}\label{sec:effects}

Let us recall the notion of polynomial ceiling, which refers to the input restriction to only polynomially long strings. We will see significant effects given by the polynomial ceilings.

\subsection{Elimination of Counters by Polynomial Ceiling Bounds}

In Section \ref{sec:reducing-counter}, we have discussed how to reduce the number of counters down to three or four.

In the presence of polynomial ceilings, it is possible to eliminate ``all'' counters from multi-counter automata and multi-counter pushdown automata.

\begin{theorem}\label{poly-ceiling}
Let $k$ be an arbitrary constant in $\nat^{+}$.
(1) $\twonct_k/\poly = \twon/\poly$. (2) $\twonpdct_k/\poly = \twonpd/\poly$. The same statements hold even if underlying nondeterministic machines are changed to deterministic ones. Moreover, the statements are true when underlying machines are all restricted to run in polynomial time.
\end{theorem}

\begin{proof}
Here, we prove only (2) since the proof of (1) is in essence similar. Since  $\twonpd \subseteq \twonpdct_k$ for all $k\geq1$, we obtain $\twonpd/\poly \subseteq \twonpdct_k/\poly$.
We next show that $\twonpdct_k/\poly \subseteq \twonpd/\poly$.
Consider any family $\LL=\{(L_n^{(+)},L_n^{(-)})\}_{n\in\nat}$ of promise problems in $\twonpdct_k/\poly$ over alphabet $\Sigma$.
There exists a polynomial $p$ satisfying $L_n^{(+)}\cup L_n^{(-)}\subseteq \Sigma^{\leq p(n)}$ for all $n\in\nat$. Take any family $\{M_n\}_{n\in\nat}$ of polynomial-size 2npdcta's that solves $\LL$ with $k$ counters in time polynomial in $(n,|x|)$. Let $q$ denote a polynomial such that, for any $n\in\nat$ and any promised input $x\in\Sigma^*$,  $q(n,|x|)$ upper-bounds the runtime of $M_n$ on $x$. Since $p$ is a ceiling, all the $k$ counters of $M_n$ hold the number at most $q(n,p(n))$ on all ``valid''  instances. Define $r(n)=q(n,p(n))$. Since $r$ is a polynomial, it is possible to express the contents of the $k$ counters in the form of inner states. Therefore, without using any counter, we can simulate $M_n$ on the input $x$ by running an appropriate 2npda whose stack-state complexity is polynomially bounded. This implies that $\LL\in\twonpd/\poly$.

The second and the third parts of the theorem follow in a similar way as described above.
\end{proof}

Theorem \ref{poly-ceiling} immediately leads to important consequences (Corollary \ref{cor-collapse}).
To state these consequence, we need the
nonuniform (i.e., the Karp-Lipton style polynomial-size advice-enhanced extension)
computational model of \emph{two-way auxiliary nondeterministic pushdown automata} (or aux-2npda's). As their variants, we further introduce
\emph{two-way auxiliary unambiguous pushdown automata} (or aux-2upda's), which are 2upda equipped with auxiliary (work) tapes. The complexity class
$\mathrm{UAuxPDA,\!\!TISP}(n^{O(1)},\log{n})/\poly$ induced by advised aux-2upda's\footnote{An advised aux-2upda's takes any promised instance of the form $(x,h(|x|))$, where $h$ is a polynomially-bounded advice function, and ends in halting states whenever it halts.}
that run in time $n^{O(1)}$ using work space $O(\log{n})$ together with polynomially-bounded advice functions.

\begin{corollary}\label{cor-collapse}
(1) If $\twodct_4 = \twonct_4$, then $\dl/\poly=\nl/\poly$. (2) If $\ptime\twodpdct_3 = \ptime\twonpdct_3$, then $\logdcfl/\poly = \logcfl/\poly$.
\end{corollary}

This corollary is obtainable from Theorem \ref{poly-ceiling} as well as the following two results: (i) $\twod/\poly = \twon/\poly$ implies $\dl/\poly=\nl/\poly$ \cite{Kap14} and
(ii) $\ptime\twodpd/\poly= \ptime\twonpd/\poly$ implies  $\logdcfl/\poly = \logcfl/\poly$ \cite{Yam21}.

In the end of this section, we draw another statement concerning the separation of nonuniform {(stack-)state} complexity classes under a working hypothesis called the \emph{linear space hypothesis} (LSH) \cite{Yam23a}. Lately, Yamakami \cite{Yam19a} introduced a nonuniform variant of this hypothesis, which is referred to as the \emph{nonuniform linear space hypothesis} (nonuniform-LSH). Let us consider a restriction of the 2CNF Boolean formula satisfiability problem (2SAT) and write $\mathrm{2SAT}_3$ for the collection of all satisfying 2CNF formulas, in which each variable appears at most 3 times in the form of literals. For example, if $\phi\equiv (x\vee y)\wedge (\bar{x}\vee z) \wedge (\bar{x}\vee \bar{y})$, then $y$ appears twice and $x$ appears exactly 3 times in the form of literals. The nonuniform linear space hypothesis\footnote{Originally, this hypothesis was defined in the setting of parameterized computing but, for simplicity reason, we here adapt it to the  non-parameterized computing setting.}
asserts that $\mathrm{2SAT}_3$ is not solvable in polynomial time using $O(n^{\varepsilon}\ell(n))$ space for any constant $\varepsilon\in[0,1)$ and a logarithmically bounded function $\ell$.

\begin{proposition}
If the nonuniform linear space hypothesis is true, then $\twodct_4\neq \twonct_4$ and $\twodpdct_3\neq \twonpdct_4$.
\end{proposition}

\begin{proof}
Assume that nonuniform-LSH is true. In \cite{Yam23a}, it was shown that , under the assumption of LSH, we obtain $\dl\neq\nl$ and $\logdcfl\neq\logcfl$. By adding advice to underlying machines, it is possible to prove that $\dl/\poly\neq\nl/\poly$ and $\logdcfl/\poly\neq \logcfl/\poly$ if nonuniform-LSH is true. The contrapositive statement of Corollary \ref{cor-collapse} then yields $\twodct_4\neq\twonct_4$ and $\ptime\twodpdct_3\neq \ptime\twonpdct_3$. The latter inequality further leads to $\twodpdct_3\neq \twonpdct_3$, as requested.
\end{proof}

\subsection{Complementation of 2N and 2NPD}\label{sec:complementary}

Let us look further into the question raised in Section \ref{sec:state-complexity} concerning the closure property of nonuniform {(stack-)state} complexity classes under complementation.
Hereafter, we particularly discuss the complementation closures of $\twon$ and $\twonpd$ when all valid instances are limited to polynomially long strings. More precisely, we wish to prove the following two equalities.

\begin{theorem}\label{complement-collapse}
(1) $\twon/\poly = \co\twon/\poly$. (2) $\twonpd/\poly = \co\twonpd/\poly$. The statement (2) still holds even if underlying 2npda's run in polynomial time.
\end{theorem}

It is important to note that the statements (1)--(2) of Theorem \ref{complement-collapse} do not require the use of parameterized versions of $\nl$ and $\logcfl$. This is rather a direct consequence of Theorems \ref{counter-four} and \ref{co-simulation} together with Theorem \ref{poly-ceiling}.
This fact exemplifies the strength of the additional use of counters provided to underlying finite automata.


\vs{-2}
\begin{proofof}{Theorem \ref{complement-collapse}}
(1) By Theorem \ref{counter-four}, $\twonct_4 = \co\twonct_4$ follows.
From this fact, by restricting every input length to be polynomially bounded, we obtain $\twonct_4/\poly = \co\twonct_4/\poly$. By Theorem \ref{poly-ceiling}(1), we then conclude that $\twon/\poly = \co\twon/\poly$.

(2) Theorem \ref{co-simulation} shows that $\twonpdct_3=\co\twonpdct_3$. We then apply Theorem \ref{poly-ceiling}(2) and obtain $\twonpd/\poly = \co\twonpd/\poly$.

The last part of the theorem also follows by a similar argument.
\end{proofof}

\subsection{Unambiguity of 2N and 2NPD}

We turn our attention to the question of whether we can make 2nfa's and 2npda's unambiguous.
In particular, we look into a relationship between $\ptime\twon$ and $\ptime\twou$ as well as between $\twonpd$ and $\twoupd$ when all promise problems are limited to having polynomial ceilings.

Unlike Theorem \ref{complement-collapse}, which relies on the direct simulation of 2npdcta's, it is not known how to simulate 2npdcta's by 2updcta's directly. Hence, we need to resort to a different strategy.

\begin{theorem}\label{NPD-vs-UPD}
(1) $\twon/\poly = \twou/\poly$. (2) $\ptime\twonpd/\poly = \ptime\twoupd/\poly$.
\end{theorem}

Theorem \ref{NPD-vs-UPD}(1) was already proven in \cite{Yam22b}. However, we do not know, for instance, whether $\twon=\twou$ or even  $\twonct_k=\twouct_k$ for any $k\geq4$ even though (1) holds.
The latter equality is compared to Theorem \ref{counter-four}, which states that $\twonct_4=\co\twonct_4$.
Hereafter, we focus on the statement (2) of Theorem \ref{NPD-vs-UPD} and provide its proof.
The basis of the proof relies on the result of Reinhardt and Allender \cite{RA00}, which asserts that $\logcfl/\poly$ coincides with $\mathrm{UAuxPDA,\!\!TISP}(n^{O(1)},\log{n})/\poly$ (which is currently expressed as $\logucfl/\poly$).


\vs{-2}
\begin{proofof}{Theorem \ref{NPD-vs-UPD}(2)}
Since unambiguity is a special case of nondeterminism, $\twoupd$ is obviously contained in $\twonpd$. Thus, $\ptime\twoupd/\poly \subseteq \ptime\twonpd/\poly$ immediately follows.
In what follows, we intend to prove the opposite inclusion under the polynomial runtime restriction, namely, $\ptime\twonpd/\poly \subseteq \ptime\twoupd/\poly$.

The proof constitutes the verification of the following claim concerning $\logucfl/\poly$ and $\ptime\twoupd$. A similar claim was proven first in \cite{Yam22a}, and later proven in \cite{Yam22b} for $\ul/\poly$ and in \cite{Yam21} for $\logcfl/\poly$. Refer to these references for their individual arguments.

\begin{yclaim}\label{2UPD-auxpda-02}
The following two statements hold.
\renewcommand{\labelitemi}{$\circ$}
\begin{itemize}\vs{-1}
  \setlength{\topsep}{-2mm}%
  \setlength{\itemsep}{1mm}%
  \setlength{\parskip}{0cm}%

\item[(i)] $\logcfl \subseteq \logucfl/\poly$ implies $\para\logcfl/\poly \cap \PHSP  \subseteq \para\logucfl/\poly$.

\item[(ii)] $\para\logcfl/\poly \cap \PHSP \subseteq \para\logucfl/\poly$ implies $\ptime\twonpd/\poly \subseteq \ptime\twoupd$.
\end{itemize}
\end{yclaim}

It is known that $\logcfl \subseteq \logucfl/\poly$ \cite{RA00}. From this fact, Claim \ref{2UPD-auxpda-02}(i) makes $\para\logcfl/\poly\cap \PHSP$ collapse to $\para\logucfl/\poly$, and Claim \ref{2UPD-auxpda-02}(ii) further yields $\ptime\twonpd/\poly \subseteq \ptime\twoupd$, as requested.

Hereafter, we wish to give the proof of Claim \ref{2UPD-auxpda-02}.


(i) Let us begin with the proof of Claim \ref{2UPD-auxpda-02}(i).
Assume that $\logcfl\subseteq \logucfl/\poly$ and let $(L,m)$ denote any parameterized decision problem in $\para\logcfl/\poly\cap\PHSP$. By definition, $m$ is log-space computable.
Since $m$ is also polynomially honest, there is a polynomial $q$ satisfying $|x|\leq q(m(x))$ for all $x$. We intend to turn $(L,m)$ into its corresponding language, say, $P$ as follows.
Since $(L,m)\in \para\logcfl/\poly$, there are a polynomially-bounded advice function $h$ and an advised aux-2npda $M_0$ such that $M_0$ recognizes the language $\{(x,h(m(x)))\mid x\in L\}$ in time $m(x)^{O(1)}$ and space $O(\log{m(x)})$.
We define $L_n^{(+)}=L\cap \Sigma_{(n)}$ and $L_n^{(-)} = \overline{L}\cap\Sigma_{(n)}$ for all indices $n\in\nat$, where $\overline{L}=\Sigma^*-L$ and $\Sigma_{(n)} = \{x\in \Sigma^*\mid m(x)=n\}$.
The desired set $P$ is then defined to be $\{(x,1^t)\mid x\in L,t\in\nat,m(x)\leq t\}$.

To show that $P\in\logcfl/\poly$, we consider the following algorithm. Let $\hat{h}(t)$ express the string $h(0)\# h(1)\# \cdots \# h(t)$.

\vs{-1}
\begin{quote}
On input $(x,1^t)$ with advice string $\hat{h}(t)$, first check whether $m(x)> t$ using log space. If so, reject the input immediately. This is possible because $m$ is log-space computable. Otherwise, run $M_0$ on $x$ with the help of advice string $h(m(x))$, which is retrieved from $\hat{h}(t)$.
\end{quote}
\vs{-1}

This algorithm can be carried out on an appropriate advised aux-2npda running in time $(t|x|)^{O(1)}$ and space $O(\log{t|x|})$. Hence, we obtain $P\in\logcfl/\poly$.

Our assumption then guarantees that $P$ belongs to $\logucfl/\poly$. Take an advice function $g$ and an advised aux-2upda $M_1$ for $P$.
From these $M_1$ and $g$, it is possible to define another advised aux-2upda $N$ that works as follows. On input $x$, compute $m(x)$ using log space and run $M_1$ on the input $(x,1^{m(x)})$ with advice string $g(m(x))$. This implies that $L$ is recognized by $N$ in time $m(x)^{O(1)}$ using space $O(\log{m(x)})$ with the help of $g$. Hence, we can conclude that  $(L,m)$ falls in $\para\logucfl/\poly$.


(ii) We then prove Claim \ref{2UPD-auxpda-02}(ii). Assume that $\para\logcfl/\poly \cap \PHSP \subseteq \para\logucfl/\poly$.
Take any family $\LL=\{(L_n^{(+)},L_n^{(-)})\}_{n\in\nat}$ of promise problems in $\ptime\twonpd/\poly$ over alphabet $\Sigma$ and
take the parameterized decision problem $(K,m)$ induced from $\LL$ via a size-incorporated family $\KK=\{(K_n^{(+)},K_n^{(-)})\}_{n\in\nat}$ of $\LL$. Notice that $K=\bigcup_{n\in\nat}K_n^{(+)}$ and $\overline{K}=\bigcup_{n\in\nat}K_n^{(-)}$.
Since $\twonpd$ is $\dl$-good as remarked earlier,
the statements (b)\&(c) of Definition \ref{def-twin} guarantee that $(K,m)\in\para\logcfl/\poly\cap\PHSP$ iff $\LL\in\ptime\twonpd/\poly$. As a consequence, we obtain $(K,m)\in\para\logcfl/\poly$.
Since $\para\logcfl/\poly\cap\PHSP \subseteq \para\logucfl/\poly$ by our assumption, $(K,m)$ belongs to $\para\logucfl/\poly$. By using Definition \ref{def-twin}(b), we can conclude that $\LL$ falls in $\ptime\twoupd/\poly$.
Therefore, Claim \ref{2UPD-auxpda-02}(ii) is true.

\s

This completes the entire proof of the theorem.
\end{proofof}

\section{Case of Unary Alphabets}\label{sec:unary-inputs}

Through Section \ref{sec:effects}, we have discussed collapses of nonuniform state complexity classes when promised inputs are all limited to polynomially long ones.
Here, we briefly mention the simplest case of unary instances in comparison with polynomial ceilings.

In the simple case of unary inputs, let us first recall the result of Geffert, Mereghetti, and Pighizzini \cite{GMP07}, who demonstrated a simulation of a ``complementary'' 2nfa (i.e., a two-way finite automaton making co-nondeterministic moves) on unary inputs by another 2nfa with a polynomial increase of the state complexity. From this fact, we can instantly draw a conclusion that $\co\twon/\mathrm{unary}$ coincides with $\twon/\mathrm{unary}$ as stated in Fig.~\ref{fig:hierarchy}. This collapse was already mentioned in \cite{Kap09}.

\begin{proposition}
$\twon/\unary = \co\twon/\unary$.
\end{proposition}

Next, we recall the result of Geffert and Pighizzini \cite{GP11}, in which  any 2nfa can be simulated by an appropriate 2ufa with a polynomial increase of the state complexity of the 2nfa. From this follows the collapse of $\twon/\mathrm{unary}$ to $\twou/\mathrm{unary}$ as stated in Fig.~\ref{fig:hierarchy}.

\begin{proposition}
$\twon/\unary = \twou/\unary$.
\end{proposition}

The above two propositions should be compared to Theorems \ref{NPD-vs-UPD}(1) and \ref{complement-collapse}(1) regarding polynomial-ceiling restrictions.

\section{Future Directions and Open Questions}

We have continued a study on the features of nonuniform {(stack-)state} complexity classes, initiated in 1978 by Sakoda and Sipser \cite{SS78}.

To obtain two major contributions of Section \ref{sec:effects} (Theorems \ref{NPD-vs-UPD} and \ref{complement-collapse}), we have employed a generic proof argument, introduced in \cite{Yam22a}, of exploiting a close connection between nonuniform state complexity classes and parameterized complexity classes as described in Section \ref{sec:parameterization}. This connection is proven in Section \ref{sec:effects} to be quite useful even if there is no (known) complete problems.

Another important contribution is to supplement multiple counters to underlying finite and pushdown automata. With the use of such counters, we have significantly expanded the scope of nonuniform {(stack-)state} complexity research in Section \ref{sec:multiple-counters}.

Here, we provide a short list of open questions for the avid reader.

\renewcommand{\labelitemi}{$\circ$}
\begin{enumerate}\vs{-1}
  \setlength{\topsep}{-2mm}%
  \setlength{\itemsep}{1mm}%
  \setlength{\parskip}{0cm}%

\item In Proposition \ref{reducing-counter}, we have reduced the number of counters used by each finite automaton down to $4$ and by each pushdown automaton to $3$. Are those numbers the minimum? If so, can we prove  that $\twonct_2\neq \twonct_3\neq \twonct_4$?

\item In Theorems \ref{counter-four} and \ref{co-simulation}, we have shown that $\twonct_4=\co\twonct_4$ and $\twonpdct_3=\co\twonpdct_3$. Can we remove any multiple counters and show that $\twon=\co\twon$ and $\twonpd=\co\twonpd$?

\item Without any input-size bounds, is it true that $\twon=\twou$ and $\twonpd=\twoupd$? Moreover, do their counter versions $\twonct_4=\twouct_4$ and $\twonpdct_3=\twoupdct_3$ hold?

\item We have seen the usefulness of polynomial ceilings. Kapoutsis \cite{Kap14} discussed the cases of non-polynomial ceilings, including $2^{O(n^k)}$ ceilings and $2^{O(\log^k{n})}$ ceilings. Can we obtain  results similar to ones in this work when all families in a given complexity class have those ceilings?

\item In Section \ref{sec:unary-inputs}, we have had a short discussion on the case of unary instances in nonuniform families of machines. It is of importance to explore the unary input case in a wider range of nonuniform state complexity classes.

\item In comparison to Theorem \ref{complement-collapse}(2), Theorem \ref{NPD-vs-UPD}(2) requires the ``polynomial runtime'' bound of underlying machines because of the use of Claim \ref{2UPD-auxpda-02}.
    Can we remove this ``polynomial runtime'' restriction from Theorem \ref{NPD-vs-UPD}(2)?
\end{enumerate}

\section*{Appendix: Proof of Lemma \ref{claim-reduction}}

In this Appendix, we provide the omitted proof of Lemma \ref{claim-reduction} from the main text for the readability.

\s

Given a family $\MM=\{M_n\}_{n\in\nat}$ of polynomial-size machines, we assume that $M_n$ runs in time $(n|x|)^{k}$ for a fixed constant $k\in\nat^{+}$.
We start with fixing a number $n\in\nat$ and $x\in\Sigma^*$ arbitrarily in the following explanation. Let $p$ denote $(n|x|)^k$ for convenience.
Now, we explain how to reduce 2 counters to one counter (with the help of 3 extra reusable counters).

For each index $j\in[2]$, $i_j$ denotes the current content of the $j$th counter of $M_n$. Let us consider the pair $(i_1,i_2)$. We encode this pair into a number defined as $i_1\cdot p + i_2$, which is succinctly expressed as $\pair{i_1,i_2}_p$. We wish to simulate the behavior of the counters of the machine using four counters CT1, CT2, CT3, and CT4. The number $\pair{i_1,i_2}_p$ will be stored in CT3, and CT2 is used to remember the value $p$ during the simulation.

(a) Here, we describe how to produce $p$ in CT2 using CT1, CT4, and a tape head.
Firstly, we store the current tape head position into CT4 by moving the tape head to the start cell
so that we will recover this tape head position after the production of $p$ in CT2. Reset CT1 and CT2 to be empty.
(i) Produce $n|x|$ in CT1 by sweeping the input tape $n$ times between the two endmarkers by the tape head. Note that $n$ is a constant for $M_n$ and it should be provided by $M_n$'s inner states.
(ii) Increment CT2 by $n|x|$ by moving the tape head. Decrement CT1 by 1.
(iii) Repeat (ii) until CT1 becomes empty.
(iv) Finally, move the tape head to the start cell and return it back to the original position by emptying CT4. Reset CT1 and CT4 to be empty.

(b) Now, we describe how to check whether $i_j=0$ or not for each $j\in[2]$ for the current value $\pair{i_1,i_2}_p$ stored in CT3. Initially, we set CT4 to be empty. Firstly, we check if CT3 is empty. If so, then we know $i_1=i_2=0$ and thus we terminate the procedure (b).
Assume that CT3 is not empty. Reset CT1 to be empty. Note that CT2 contains $p$ by the procedure (a).  In the following process, we maintain the condition that the sum of the contents of CT3 and CT4 equals $\pair{i_1,i_2}_p$.  Similarly, we keep the condition that the sum of the contents of CT1 and CT2 equals $p$.
(i) Decrement CT2 and CT3 and increment CT1 and CT4 simultaneously by one until either CT2 or CT3 becomes empty.
(iii) If CT3 is empty, then we know that $i_1=0$. In this case, $i_2$ must be positive.
Move the content of CT4 to CT3 by decrementing CT4 and incrementing CT3  simultaneously until CT4 is empty. Reset CT1 and CT2 to be empty.
(iv) If CT2 is empty, then we know that $i_1\neq0$. Move the content of CT1 to CT2. Repeat (ii)-(iv) until (ii) stops with CT3 being empty. Note that $i_1$ equals the total number of repetitions.
(v) Finally, recover $\pair{i_1,i_2}_p$ in CT3 by moving the content of CT4 to CT3.

(c) Let us describe how to simulate a push/pop operation of the counters made by $M_n$. Assume that $(i_1,i_2)$ is changed to $(j_1,j_2)$ by $M_n$ in a single step. After the procedure (b), we already know whether $i_1$ as well as $i_2$ is either zero or positive. (i) If $j_1=i_1+1$, then keep incrementing CT3 and CT1 as we decrement CT2 down to $0$. Move the content of CT1 to CT2 by emptying CT1. (ii) If $j_2=i_2+1$, then we push ``$1$'' into CT3.
(iii) If $i_1>0$ and $j_1=i_1-1$, then we decrement CT1, CT2, and CT3 until CT2 becomes empty. Add the remaining content of CT1 to CT2.
(iv) If $i_2>0$ and $j_2=i_2-1$, then we pop ``$1$'' from CT3.

(d) If a stack is usable during the processes (a)--(c), then we can use the stack as a counter by first writing a separator, say, $\#$ used as a new bottom marker and starting either pushing ``$1$'' or popping it.  This makes it possible to reduce the number of counters further down to three.

This completes the proof of the lemma.

\let\oldbibliography\thebibliography
\renewcommand{\thebibliography}[1]{%
  \oldbibliography{#1}%
  \setlength{\itemsep}{-2pt}%
}
\bibliographystyle{alpha}


\end{document}